\documentclass{article}
\usepackage{amsthm, amsmath}
\usepackage[mathscr]{eucal}
\usepackage{amssymb}
\usepackage{microtype}
\usepackage{graphicx}
\usepackage{caption}
\usepackage{float}
\usepackage{multirow}
\usepackage[numbers]{natbib}
\usepackage[ruled,vlined]{algorithm2e}
\usepackage{algorithmic}
\usepackage{wrapfig}

\usepackage[T3,T1]{fontenc}
\DeclareSymbolFont{tipa}{T3}{cmr}{m}{n}
\DeclareMathAccent{\invbreve}{\mathalpha}{tipa}{16}

\newcommand{\symmdiff}{\bigtriangleup} 

 \newcommand{\lovasz}{Lov\'asz}

\newcommand{\ggrow}{\hat{g}} 
\newcommand{\gshrink}{\check{g}}  
\newcommand{\gbar}{\bar{g}} 
\newcommand{\growsymb}[1]{\hat{#1}} 
\newcommand{\shrinksymb}[1]{\check{#1}}
\newcommand{\barsymb}[1]{\bar{#1}}
\newcommand{\inpolygrow}{\ensuremath{\hat{\partial}^f}}
\newcommand{\inpolyshrink}{\ensuremath{\check{\partial}^f}}
\newcommand{\inpolybar}{\ensuremath{\bar{\partial}^f}}
\newcommand{\inpolygrowshrink}{\ensuremath{\mathring{\partial}^f}}

\DeclareMathOperator*{\argmax}{argmax}
\DeclareMathOperator*{\argmin}{argmin}

\newcommand{\lex}[1]{{\ensuremath \breve #1}}
\newcommand{\mex}[1]{{\ensuremath \tilde #1}}
\newcommand{\cex}[1]{{\ensuremath \invbreve #1}}

\newcommand{\perm}{\ensuremath{\sigma}}

\usepackage[usenames,dvipsnames,dvinames]{xcolor}
\providecommand{\dodraft}{false}
\newboolean{isdraft}
\setboolean{isdraft}{\dodraft} 
\ifthenelse{\boolean{isdraft}}{%

\newcommand{\rishabh}[2]{{\color{blue}{#1 $\to$ {\bf Rishabh}}: #2}}
\newcommand{\jeff}[2]{{\color{red}{#1 $\to$ {\bf Jeff}}: #2}}

} {

\newcommand{\rishabh}[2]{}
\newcommand{\jeff}[2]{}

}

\newtheorem{theorem}{Theorem}[section]
\newtheorem{lemma}[theorem]{Lemma}
\newtheorem{proposition}[theorem]{Proposition}
\newtheorem{observation}{Observation}[section]
\newtheorem{corollary}[theorem]{Corollary}

\newtheorem{example}{Example}[section]

\usepackage[usenames,dvipsnames,dvinames]{xcolor}
\usepackage[bookmarks, colorlinks=true, citecolor=Violet,linkcolor=Mahogany,urlcolor=blue]{hyperref}

\begin{document}
\title{Polyhedral aspects of Submodularity, Convexity and Concavity}
\author{Rishabh Iyer \\
Dept. of Electrical Engineering\\  
University of Washington\\ 
Seattle, WA-98175, USA
\and Jeff Bilmes \\
Dept. of Electrical Engineering\\  
University of Washington\\ 
Seattle, WA-98175, USA}
\maketitle

\begin{abstract}
  The seminal work by Edmonds~\cite{edmondspolyhedra} and
  \lovasz{}~\cite{lovasz1983} shows the strong connection between
  submodular functions and convex functions. Submodular functions have
  tight modular lower bounds, and a subdifferential
  structure~\cite{fujishige1984subdifferential} in a manner akin to
  convex functions. They also admit polynomial time algorithms for
  minimization and satisfy the Fenchel duality
  theorem~\cite{fujishige1984theory} and the discrete separation
  theorem~\cite{frank1982algorithm}, both of which are fundamental
  characteristics of convex functions. Submodular functions also have
  properties similar to concavity. For example, submodular function
  maximization, though NP hard, admits constant factor approximation
  guarantees. Concave functions composed with modular functions are
  submodular, and they also show the diminishing returns property. In
  this manuscript, we try to provide a more complete picture on the
  relationship between submodularity and both convexity and concavity
  --- we do this by extending many of the results connecting
  submodularity with convexity~\cite{lovasz1983, frank1982algorithm,
    fujishige1984theory, edmondspolyhedra,
    fujishige1984subdifferential} to the concave aspects of submodular
  functions.  We first show the existence of superdifferentials (a
  polyhedral partitioning of $\mathbb R^V$) and efficiently computable
  tight modular upper bounds of a submodular function. While we show
  that it is hard to characterize these polyhedra, we obtain inner and
  outer bounds on the superdifferential along with certain specific
  and useful supergradients. We then investigate forms of concave
  extensions of submodular functions and show interesting
  relationships to submodular maximization. We next show connections
  between optimality conditions over the superdifferentials and
  submodular maximization, and show how forms of approximate
  optimality conditions translate into approximation factors for
  maximization. We end this paper by studying versions of a ``concave'' discrete
  separation theorem and the Fenchel duality theorem when seen from
  the concave point of view. In every case, we relate our results to
  the existing results from the convex point of view, thereby
  improving the analysis of the relationship between submodularity,
  convexity, and concavity.
\end{abstract}

\section{Introduction}
\label{sec:introduction}

Long known to be an important property for problems in
combinatorial optimization, economics, operations research, and game
theory, submodularity is gaining popularity in a number of new areas including machine learning. Along with its natural connection to many application domains, it also admits a number of interesting theoretical characterizations. A function $f: 2^V \to \mathbb R$ over a ground set $V = \{1, 2,
\cdots, n\}$ is \emph{submodular} if for all subsets $S, T \subseteq
V$, it holds that, 
\begin{align}
f(S) + f(T) \geq f(S \cup T) + f(S \cap
T). 
\end{align}
Equivalently, a submodular set function satisfies
\emph{diminishing marginal returns}: Define $f(j | S) \triangleq f(S \cup \{ j \}) -
f(S)$ as the marginal cost of element $j\in V$ with respect to $S
\subseteq V$.\footnote{We also use this notation for sets $A,B$ as in $f(A|B) = f(A\cup B) - f(B)$.}
The diminishing returns property states that, 
\begin{align}
f(j | S) \geq f(j | T), \forall S \subseteq T \text{ and } j \notin T.
\end{align} 
Through the rest of the paper below, we shall also assume without loss of generality that $f(\emptyset) = 0$.

\paragraph{Submodularity and convexity:}
Submodular functions have been strongly associated with convex functions, to the extent that submodularity is sometimes regarded as a discrete analogue
of convexity~\cite{fujishige2005submodular}. 
This relationship is evident by the fact that submodular function
minimization is easy in that there exist strongly polynomial time
algorithms which achieve it. This is akin to convex minimization which
is also easy. A number of recent results, however, make this
relationship much more
formal. 
For example, similar to convex functions, submodular functions have
tight modular lower bounds and admit a subdifferential
characterization~\cite{fujishige1984subdifferential}. Moreover, it is
possible~\cite{fujishige1984theory} to provide optimality conditions,
in a manner analogous to the Karush-Kuhn-Tucker (KKT) conditions from
convex programming, for submodular function minimization. Furthermore,
the Fenchel duality theorem and the discrete separation theorem, both
of which are known to hold for convex functions have been shown to 
hold also for submodular functions~\cite{fujishige1984theory,
  frank1982algorithm}. Submodular functions also admit a
natural convex extension, known as the \lovasz{} extension, that is
easy to evaluate~\cite{lovasz1983} and optimize. The \lovasz{}
extension, moreover, also has no integrality gap and minimizing a
submodular function is equivalent to minimizing its \lovasz{}
extension. All these results show that submodularity is indeed
closely related to convexity, and seems to verify the claim
that submodularity is ``the'' discrete analog of convexity.

\paragraph{Submodular functions and concavity:} Submodular functions
also have properties that are unlike convexity and are more akin to
concavity. Submodular function maximization is known to be NP
hard. However, there exist a number of constant factor approximation
algorithms based on simple greedy or local search
heuristics~\cite{janvondrak, lee2009non, nemhauser1978} and some
recent continuous approximation methods~\cite{chekuri2011submodular,
  feldman2011unified}. This is unlike convexity where maximization can
be hopelessly
difficult~\cite{sahni1974computationally}.\rishabh{Jeff}{add a
  citation to where this is proven.}\jeff{Rishabh}{Done.}
Furthermore, submodular functions have a diminishing returns property
which is similar to concavity, and concave over modular functions are
known to be submodular. In addition, submodular functions have been
shown to have tight modular upper
bounds~\cite{rkiyersemiframework2013, rkiyersubmodBregman2012,
  rkiyeruai2012, jegelkacvpr, jegelkanips}, and as we show, possess
superdifferentials and supergradients very much like concave
functions. The multi-linear extension of a submodular function, which
is useful~\cite{chekuri2011submodular} for example in the context of
submodular maximization, is known to be concave when restricted to a
particular direction. All these seem to indicate that submodular
functions are related both to convexity and to concavity. In some
sense, submodular functions are strange and lucky --- convex and
concave functions each have distinct and useful properties, while
submodular functions have best of both worlds.  In this paper, we
formalize these relationships.\looseness-1

\begin{figure}
\centering
\includegraphics[page=1,width = 0.4\textwidth]{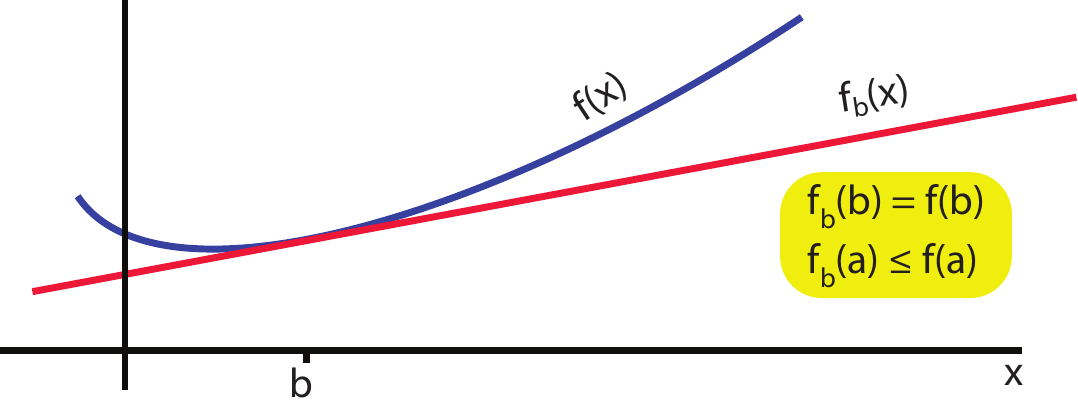}
~~~~~~
\includegraphics[page=2,width = 0.4\textwidth]{subgradsupergradconvexconcave.pdf}
\caption{Convex and Concave functions with sub and super gradients respectively} 
\end{figure}

\subsection{Motivation and Past Work}

For more than four decades, researchers have been investigating
theoretical and algorithmic aspects of submodular functions. The bulk
of this work~\cite{fujishige2005submodular, edmondspolyhedra,
  lovasz1983, fujishige1984subdifferential, frank1982algorithm,
  fujishige1984submodular} has been in relating submodular functions
to convexity from a polyhedral perspective, thereby culminating in
efficient algorithms for submodular minimization. From a polyhedral
perspective, Fujishige, Edmonds and others
\cite{fujishige2005submodular, edmondspolyhedra,
  fujishige1984subdifferential}, provided a characterization of the
submodular polyhedron, the base polytope, and subdifferentials of
submodular functions.  \lovasz~\cite{lovasz1983} then provided an
efficient characterization of the convex extension of a submodular
function, which has become known as the \lovasz{} extension and also
the Choquet integral~\cite{choquet1953theory}.  The connection between
submodularity and convexity was made still more precise when it was
shown~\cite{frank1982algorithm, fujishige1984submodular} that the
discrete separation theorem, Fenchel duality theorem, and the
Minkowski sum theorems hold for submodular functions, when seen as
analogous to convexity. From a computational perspective, these results
have helped provide several algorithms for submodular function
minimization. In particular, \cite{fujishige2005submodular, frbach1}
use the submodular polyhedron and the convex extension to provide an
exact algorithm for submodular minimization. Similarly,
\cite{schrijver2000combinatorial,iwata2001combinatorial, iwata2003faster,orlin2009faster,iwata2009simple} and others have used
many of these ideas to provide exact algorithms for submodular
minimization.
 
While submodular functions are related to concavity (as discussed
above), the polyhedral aspects of submodular functions from the
perspective of maximization (and that we address in this paper) have
not been nearly as well studied.  Most work on submodular maximization
has been on the exploration of approximation algorithms.  The first
set of results for submodular maximization were shown
in~\cite{nemhauser1978, nemhauser78}, where they provide a $1 - 1/e$
approximation algorithm (in the form of a simple greedy heuristic) for
maximizing a monotone submodular function under a cardinality
constraint. Further variants of the greedy algorithm were also
extended to matroid and knapsack constraints~\cite{fisher1978analysis,
  sviridenko2004note, krause2005note, linbudget}. The factor $1 - 1/e$
was also shown to be optimal under the value oracle
model~\cite{feige1998threshold, nemhauser78}.  The first systematic
study on non-monotone submodular function maximization was performed
by Fiege \textit{et al}~\cite{janvondrak}, where they obtain a $1/3$
and a randomized $2/5$ approximation for unconstrained submodular
maximization. They also show an absolute hardness of $1/2$ for this
problem.  They raised an open question, however, whether there exists
a tight $1/2$ approximation algorithm for this problem. This question
was resolved in~\cite{feldman2012optimal}, where they show that a
simple randomized linear time algorithm achieves an approximation
factor of $1/2$ in expectation.  Many of these results can be extended
to matroid and knapsack constraints in~\cite{lee2009non,
  matroidimproved}.


Polyhedral aspects of submodular maximization and the concave
extension of a submodular function have been studied but only in a relatively limited
context~\cite{janvondrak, vondrakcontinuousgreedy,
  vondrak2007submodularity, dughmi2009submodular, nemhauser1978,
  rkiyersubmodBregman2012, rkiyeruai2012, jegelkacvpr, jegelkanips,
  Boros2002155}.  \rishabh{Jeff}{there might be some earlier refs.\ to
  concave extensions in general in the pseudo-Boolean function
  literature, see Boros and Hammer.}\jeff{Rishabh}{Done.}  For
example, a recent chain of work by Jan Vondr{\'a}k and
others~\cite{vondrakcontinuousgreedy, vondrak2007submodularity,
  dughmi2009submodular} investigated concave extensions of a
submodular function, which were shown to be NP hard to evaluate
\cite{vondrak2007submodularity}. Similarly the 
submodular
semidifferentials has
gained a lot of attention from the machine learning community.  
In
particular, the subgradients and supergradients of a submodular
function have inspired a unifying Majorization-Minimization framework
for submodular
optimization~\cite{narasimhanbilmes,rkiyersubmodBregman2012,
  rkiyeruai2012, jegelkacvpr, jegelkanips, nipssubcons2013,
  curvaturemin}. These semidifferentials have also been used in the
context of approximate inference in a class of probability
distributions defined via submodular functions~\cite{SPP2014discml,
 djolonga14map}, and have also been used to define a class of Bregman
divergences using submodular functions~\cite{rkiyersubmodBregman2012}.

In this paper, we attempt to provide a first unifying characterization of the concave aspects of submodular functions from a polyhedral perspective, thereby extending many of the observations made in~\cite{lovasz1983}. In this effort, we discover a number of interesting connections between these different aspects of submodular functions connecting concavity, and contrast them to known results of submodularity and convexity.

\subsection{Our Contributions}
The main contributions of this work is in providing the first systematic theoretical study related to polyhedral aspects of submodular function maximization and connections to concavity. The following provides a summary of the main components and contributions of this paper.
\begin{itemize}
\item We show that submodular functions have tight modular (additive) upper bounds, thereby proving the existence of the superdifferential of a submodular function. We show that characterizing this subdifferential is NP hard in general. However, we provide a series of (successively tighter) outer and also inner polyhedral bounds, all obtainable in polynomial time, and also show that we can obtain some specific practically useful supergradients in polynomial time.
Along the way, we relate this to $M^{\natural}$-concave submodular functions 
\cite{murota2003discrete}
defined on $2^V$.
 
\item We also extend the notion of the submodular polyhedron (which
  consists of the set of modular lower bounds of a submodular
  function, and for reasons that will become clear, we will refer to
  as the ``submodular {\bf lower} polyhedron''). We then define the submodular
  {\bf upper} polyhedron (which consists of the set of modular upper bounds of the
  submodular function).

\item We define the concave extension of a submodular function, in a
  manner similar to the convex extension, namely as a linear program
  over the submodular upper polyhedron. We show that this is
  identical to the concave extensions considered in the
  past~\cite{vondrakcontinuousgreedy, vondrak2007submodularity}. We
  also provide a family of concave extensions based on bounds on the
  submodular upper polyhedra, some of which can be efficiently
  computed in polynomial time. We relate these extensions to
  submodular function maximization.

 \item We then show how we can define forms of 
 optimality conditions
\jeff{Jeff}{I think that optimality conditions are more general than KKT, so is
it exact to call them KKT-like conditions or perhaps more generally, just
``optimality conditions''?}\rishabh{Jeff}{this comment is for you too.}
for submodular maximization through the submodular superdifferential. We also show how optimality conditions related to approximations to the superdifferential lead to a number of familiar approximation guarantees for these problems.\looseness-1
 
\item Finally we study the Fenchel duality and discrete separation
  theorems for submodular functions seen in connection to
  concavity. While in general this does not hold, we show that these
  hold under certain quite mild conditions.  We also show how the
  Minkowski-Sum theorem also holds under certain restricted
  conditions.\rishabh{Jeff}{I'm leaving this in for now, lets talk
    about it.}\jeff{Rishabh}{This is now a part of the paper.}

\item Throughout this paper, we point to interesting connections
  regarding how our results generalize many of the results of
  $M^{\natural}$-concave submodular
  functions~\cite{murota2003discrete} on $2^V$, where many of these
  characterizations are exact.
 \end{itemize}

\rishabh{Jeff}{
(3/31/2015: Note the following is a comment from 2/23/2015, from the pdf, but I'm adding it here. Is this now addressed?)

In the abstract, or the intro, mention very precisely and explicitly how this work generalizes some of the results that existed for M-natural convex functions of Murota (i.e,. that is special case of submodular function where a lot of the results hold exactly) but here where it holds approximately.

I.e., right away, make sure that it is clear how this generalizes over the Murota work when things do not hold exactly,b ut how they hold approximately (and how one can from a superdifferential perspective get many of the existing approximation bounds relating to submodular max).
}\jeff{Rishabh}{I added a bit on this as the last item above.}

\rishabh{Jeff}{Another case to be made for the utility of the paper is the following, which I added as motivation during the talk. In big-data continuous
optimization problems, gradient based methods have become the only practical
scalable way to perform optimization, as perhaps they are most efficient
in terms of work done per flop. While we've already made use
of these gradients for optimization, starting with the SSP, and with
the discrete mirror descent stuff, one hope for the current paper
is that by expanding the study of semidifferentials of submodular functions,
still new optimization strategies based perhaps on a mirror-descent-like
frame work that, even if they don't' have the very best approximation guarantee,
might be quite useful because of their practicality and scalability. After
all, in the continuous optimization world, the stochastic approximation
methods developed for convex functions tend to work well also for non-convex
functions. Perhaps some of the mirror-descent methods based on semigradients
would end up doing well in practice even for non-submodular functions which, as you know, are also quite important. I think we should turn the above message into text for the paper. I can work on that in a next pass.}\jeff{Rishabh}{Yes, I think that's a good idea.}

\subsection{Road-Map of this paper}
In Sections~\ref{polysubconvex}, \ref{submodcvxext},
\ref{submodminopt} and \ref{submodfdtdstmin}, we review the
connections between submodularity and convexity. Most of the results
in these sections are from~\cite{lovasz1983, fujishige2005submodular},
and in some cases we provide some generalizations. In
Section~\ref{polysubconvex}, we review polyhedral aspects of
submodularity and convexity, and investigate the submodular
polyhedron, submodular subdifferentials, etc. In
Section~\ref{submodcvxext}, we study the convex extensions of a
submodular function, while in Section~\ref{submodminopt} we review the
optimality conditions of submodular function minimization from a
polyhedral perspective. In Section~\ref{submodfdtdstmin}, we review
the \emph{discrete separation theorem}, the Fenchel duality theorem,
and the Minkowski sum theorem, all from the perspective of the convex
analogy of submodular functions.  In Section~\ref{polysubconcave} we
define and investigate the polyhedral aspects of submodularity and
concavity --- we do this by defining the submodular upper polyhedron
and the submodular superdifferentials. In Section~\ref{submodccvext},
we provide a characterization of the concave extension of a submodular
function. In Section~\ref{submodmaxopt}, we study the optimality
conditions of submodular function maximization from a polyhedral
perspective. Finally, in Section~\ref{submodfdtdstmax}, we provide
versions of the discrete separation theorem, the Fenchel duality
theorem and the Minkowski sum theorem but from the perspective of
concavity of a submodular function.

\section{Polyhedral aspects of Submodularity and Convexity}
\label{polysubconvex}

Most of the results in this section are covered
in~\cite{edmondspolyhedra,lovasz1983, fujishige2005submodular} and the
references contained therein, so for more details please refer to
these texts. We use this section to review existing work on the
polyhedral connections between submodularity and convexity and to help
contrast these with the corresponding results on the polyhedral
connections between submodularity and concavity starting in
Section~\ref{polysubconcave}.  \looseness-1

\subsection{Submodular (Lower) Polyhedron}
\label{submodpoly}

\begin{figure}
\centering
\includegraphics[page=2,width = 0.3\textwidth]{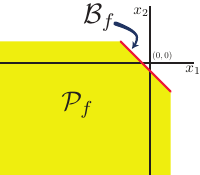}
\caption{The Submodular Polyhedron $\mathcal P_f$ and the Base Polytope $\mathcal B_f$ in two dimensions}
\end{figure}

For a submodular function $f$, the submodular (lower)
polyhedron\footnote{Since the submodular polyhedron consists of
  modular lower bounds of a submodular function, we shall also call it
  the submodular \emph{lower} polyhedron to contrast with
the submodular \emph{upper}  polyhedron
we introduce in Section~\ref{uppersubpolysec}.} and the base polytope of a
submodular function \cite{fujishige2005submodular} are defined,
respectively, as:
\begin{align}
\mathcal P_f \triangleq \{ x \in \mathbb R^V : x(S) \leq f(S), \forall S \subseteq V \} 
\qquad
\text{ and }
\qquad
\mathcal B_f \triangleq \mathcal P_f \cap \{ x \in \mathbb R^V : x(V) = f(V) \},
\end{align}
where $x(S) = \sum_{i \in S} x_i$ for any $S \subseteq V$. The submodular polyhedron has a
number of interesting properties, one important one being that the
extreme points and facets can easily be
characterized even though the polyhedron itself is described by a
exponential number of inequalities. In fact, surprisingly, every
extreme point of the submodular polyhedron is an extreme point of the
base polytope. These extreme points admit an interesting
characterization in that they can be computed via a simple greedy
algorithm \cite{edmondspolyhedra} --- let $\perm$ be a permutation of $V = \{1, 2, \cdots,
n\}$. Each such permutation defines a chain with elements $S^{\perm}_0
= \emptyset$, $S^{\perm}_i = \{ \perm(1), \perm(2), \dots, \perm(i)
\}$ such that $S^{\perm}_0 \subseteq S^{\perm}_1 \subseteq \cdots
\subseteq S^{\perm}_n$. This chain defines an extreme point
$h^{\perm}$ of $\mathcal P_f$ with entries
\begin{align}
\label{eq:permmod}
h^{\perm}(\perm(i)) = 
f(S^{\perm}_i) - f(S^{\perm}_{i-1}). 
\end{align}
Each permutation of $V$ characterizes an extreme point of $\mathcal P_f$ and all possible extreme points of $\mathcal P_f$ can be characterized in this manner~\cite{fujishige2005submodular}. Furthermore, the problem $\max_{y \in \mathcal P_f} y^{\top} x$,
which is a linear program over a submodular polyhedron, can be very efficiently computed through the greedy algorithm~\cite{edmondspolyhedra}. The following lemma gives the greedy algorithm for finding this.
\begin{lemma}~\cite{edmondspolyhedra, lovasz1983} \label{greedy}
Given a vector $w \in \mathbb{R}^n_{+}$, consider a permutation $\perm_w$, such that $w[\perm_w(1)] \geq w[\perm_w(2)] \geq \cdots \geq w[\perm_w(n)]$. Define $s^*(\perm_w(i)) = f(S^{\perm_w}_i) - f(S^{\perm_w}_{i-1})$ for $i \in \{1, 2, \cdots, n\}$. Then $\argmax_{s \in P_f} w^{\top} s = \argmax_{s \in B_f} w^{\top} s\ni s^*$.
Furthermore, $\max_{s \in P_f} w^{\top} s = \sum_{i = 1}^n w(\perm_w(i)) [f(S^{\perm_w}_i) - f(S^{\perm_w}_{i-1}]$
\end{lemma}
It is immediate that the optimizers $s^*$ above form extreme points of the submodular polyhedron. 
Also, given a submodular function $f$ such that\footnote{Any set function 
$h$ is said to be {\em normalized} if $h(\emptyset) = 0$.} $f(\emptyset) = 0$, the condition that $x \in \mathcal P_f$ can be checked in polynomial time for every $x$ --- this follows directly from the fact that submodular function minimization is polynomial time.
\begin{proposition}
Given a submodular function $f$, checking if $x \in \mathcal P_f$ is equivalent to the condition $\min_{X \subseteq V} [f(X) - x(X)] \geq 0$, which can be checked in poly-time.
\end{proposition}

\subsection{The Submodular Subdifferential}
\label{sec:subm-subd}

Another aspect of the connection between submodular functions and
convexity is the submodular
subdifferentials~\cite{fujishige1984subdifferential}. The
subdifferential $\partial_f(X)$ of a submodular set function $f: 2^V
\to \mathbb{R}$ for a set $X \subseteq V$ is defined
\cite{fujishige1984subdifferential, fujishige2005submodular}
analogously to the subdifferential of a continuous convex function:
\begin{align}
\label{submodsubdiff}
\partial_f(X) &\triangleq \{x \in \mathbb{R}^n: f(Y) - x(Y) \geq f(X) - x(X)\;\text{for all } Y \subseteq V\} 
\end{align}
The polyhedra above can be defined for any (not necessarily submodular) set function. When the function is submodular however, it can be characterized efficiently.
\begin{figure}
\centering
\includegraphics[page=1,width = 0.3\textwidth]{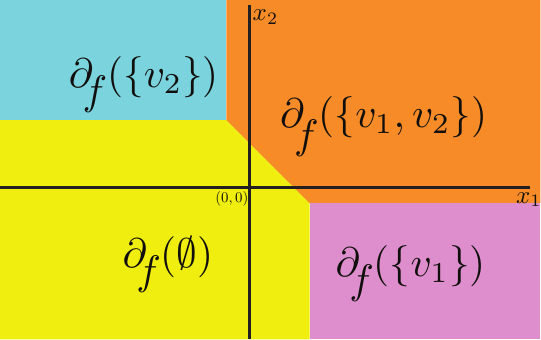}
\caption{The Subdifferentials $\partial_f(Y)$ of a submodular function for different sets $Y$ in two dimensions. Notice that the subdifferentials partition the space $\mathbb{R}^2$. In this case, $V = \{v_1, v_2\}$.}
\end{figure}
Firstly, note that for normalized submodular functions, for any $h_X
\in \partial_f(X)$, we have $f(X) - h_X(X) \leq 0$ which follows by
the constraint at $Y=\emptyset$.  Like the submodular polyhedron,
the extreme points of the submodular subdifferential also admit
interesting characterizations. We shall denote a subgradient at $X$ by
$h_X \in \partial_f(X)$. Similar to the submodular polyhedron, the
extreme points of $\partial_f(X)$, for any $X$, may be computed via a greedy
algorithm as follows: let $\perm$ be a permutation of $V$ that assigns the
elements in $X$ to the first $|X|$ positions ($i \leq |X|$ if and only
if $\perm(i) \in X$) and $S^{\perm}_{|X|} = X$. An illustration of
this is shown in Figure~\ref{chainim}.

\begin{figure}[h]
\centering
\includegraphics[width = 0.3\textwidth]{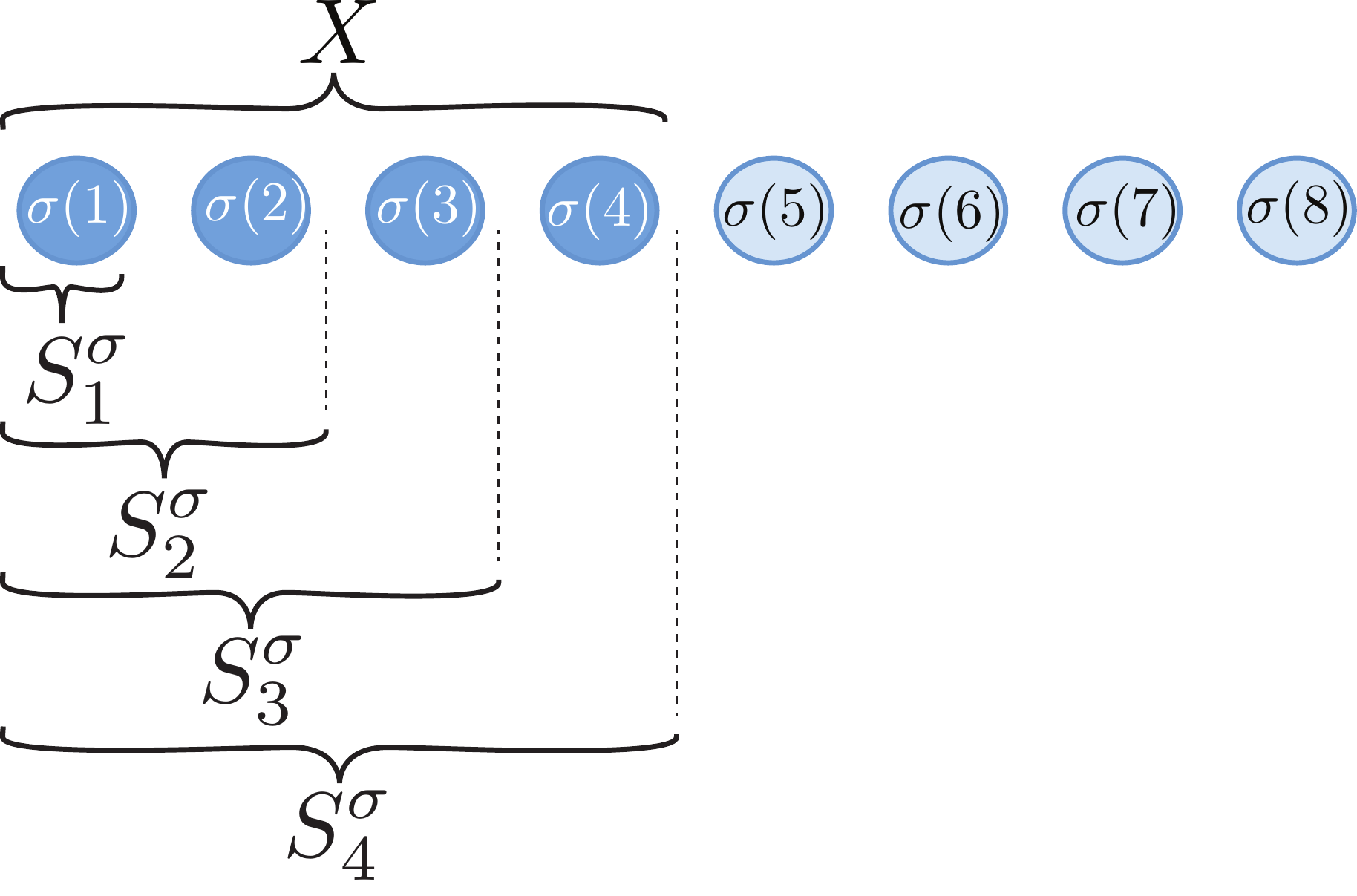}
\caption{A visualization of 
a permutation $\perm = (\perm(1), \perm(2), \dots)$ of $V$
and the chain of sets $S_1^\perm, S_2^\perm, \dots$
according to this permutation, where $S_i^\perm = \{ \perm(1), \perm(2), \dots, \perm(i) \}$. 
Here, also, we show the permutation is compatible with $X = S_{|X|}^\perm$ with $|X|=4$.}
\jeff{Jeff}{redo this figure, I think I have one from my class which we can use.}\jeff{Rishabh}{I have changed this figure a bit.}
\rishabh{Jeff}{we should avoid using bitmap figures, rather try to always use vector figures for any graphic since they print and scale much better. I've done this now for this fig but lets do that in future, ok? thanks.}\jeff{Rishabh}{Sure.}
\label{chainim}
\end{figure}

This chain defines an extreme point
$h^{\perm}_X$ of $\partial_f(X)$ with entries
\begin{align}
h^{\perm}_X(\perm(i)) = 
f(S^{\perm}_i) - f(S^{\perm}_{i-1}). 
\end{align}

Note that for every subgradient $h_X \in \partial_f(X)$ we can define
a modular function
\begin{align}
m_X(Y) \triangleq f(X) + h_X(Y) - h_X(X)
\label{eq:mod_tight_lower_bound}
\end{align}
that is defined $\forall Y \subseteq V$, and that is a tight lower
bound of $f$ --- that is, $m_X$ satisfies $m_X(Y) \leq f(Y), \forall Y
\subseteq V$ and we have that $m_X(X) = f(X)$.  Hence, the
subdifferential corresponds exactly to the set of tight modular lower
bounds of a submodular function, at a given set $X$. If we choose
$h_X$ to be an extreme subgradient, the modular lower bound becomes
$m_X(Y) = h_X(Y)$, resulting in a normalized modular function (i.e.,
$m_X(\emptyset) = 0$).  Also, if $X = S^\sigma_j$ for some $j$, 
then since $h^\sigma_X(S^\sigma_i) =
f(S^\sigma_i)$ for all $i$, the modular lower bound defined as
$m^\sigma_X(Y) = h^\sigma_X(Y)$ has the property that it is tight for
all sets $\{ S^\sigma_i \}_i$, not just $X$.

The subdifferential defined in Eqn.~\eqref{submodsubdiff} is defined via an exponential number of inequalities. A key observation however is that many of these inequalities are redundant. We define three polyhedra:
\begin{align}
\partial_f^1(X) &\triangleq \{x \in \mathbb{R}^n: f(Y) - x(Y) \geq f(X) - x(X), \forall Y \subseteq X \} \\
\partial_f^2(X) &\triangleq \{x \in \mathbb{R}^n: f(Y) - x(Y) \geq f(X) - x(X), \forall Y \supseteq X \} \\
\partial_f^3(X) &\triangleq \{x \in \mathbb{R}^n: f(Y) - x(Y) \geq f(X) - x(X), \forall Y: Y \not \subseteq X, Y \not \supseteq X\}
\end{align}
We immediately have that $\partial_f(X) = \partial_f^1(X)
\cap \partial_f^2(X) \cap \partial_f^3(X)$. The following lemma shows
that the inequalities in $\partial_f^3(X)$ are redundant in
characterizing $\partial_f(X)$ when given $\partial_f^1(X)$ and
$\partial_f^2(X)$.
\begin{lemma}(\cite[Lemma 6.4]{fujishige2005submodular})
Given a submodular function $f$, $\partial_f(X) = \partial_f^1(X) \cap \partial_f^2(X)$. Hence,
\begin{align}
\label{subdiffred} 
\partial_f(X) = \{x \in \mathbb{R}^n: f(Y) - x(Y) \geq f(X) - x(X), \forall Y \in [\emptyset, X] \cup [X, V]\} 
\end{align}
\end{lemma}
In the above, $[A, B] = \{X \subseteq V: A \subseteq X \subseteq B\}$
whenever $A \subseteq B$. We thus see that for $X \neq \{\emptyset, V\}$,
many of the inequalities defining $\partial_f(X)$ 
in Eqn.~\eqref{submodsubdiff}
are in fact
redundant. 

The subdifferential at the emptyset has a special relationship since
$\partial_f(\emptyset) = \mathcal P_f$.  Similarly
$\partial_f(V) = \mathcal P_{f^{\#}}$, where
$f^{\#}(X) = f(V) - f(V \backslash X)$ is the submodular dual of $f$.
Furthermore, since $f^{\#}$ is a supermodular function, it holds that
$\partial_f(V)$ is a \emph{supermodular polyhedron} (for a
supermodular function $g$, the supermodular polyhedron
is defined as
$\mathcal P_g = \{x \in \mathbb R^V: x(X) \geq g(X), \forall X \subseteq V\}$).
\rishabh{Jeff}{I added back in this bit since I think it might be
  important to point out a distinction between the submodular upper
  polyhedron and the supermodular polyhedron since they have similar
  looking definitions. When I mentioned this during the talk, I
  noticed quite a few people nodding in understanding.}\jeff{Rishabh}{Sure}

The following lemma shows another instructive fact about the subdifferentials:
\begin{lemma}(\cite[Lemma 6.5]{fujishige2005submodular})\label{dirprod}
For any submodular function $f$, $\partial_f(X) = \partial_{f^X}(X) \times \partial_{f_X}(\emptyset)$, where $f^X(Y) = f(Y), \forall Y \subseteq X$, and $f_X(Y) = f(Y \cup X) - f(X), \forall Y \subseteq V \backslash X$, and $\times$ denotes the direct product.
\jeff{Jeff}{add a bit more detail here.}
\end{lemma}

Finally we define what we call the local approximation of the
subdifferential as follows:
\begin{align}
\label{localsub}
\partial_f^{\symmdiff(1, 1)}(X) 
\triangleq \{x \in \mathbb{R}^V : \forall j \in X, f(j | X \backslash j) \leq x(j) \text{ and } \forall j \notin X, f(j | X) \geq x(j)\}. 
\end{align}
Notice that $\partial_f^{\symmdiff(1, 1)}(X) \supseteq \partial_f(X)$
since we have fewer constraints here than in the original subdifferential. In
particular $\partial_f^{\symmdiff(1, 1)}(X)$ considers only $n$
inequalities by choosing the sets $Y$ in Eqn.~\eqref{subdiffred} such
that $|Y \symmdiff X| = 1$ (i.e., Hamming distance one away from $X$). This polyhedron
will be useful in characterizing local minimizers of a submodular
function (see Section~\ref{submodminopt}) and motivating analogous
constructs for local maxima (see, for example,
Proposition~\ref{prop:local_max_poly}).

\subsection{Generalized Submodular Lower Polyhedron}
\label{gensubmodpoly}

\rishabh{Jeff}{I think below we are defining the generalized polyhedron
for non-submodular functions so this section should clarify that it is
defining the 
``Generalized Lower Polyhedron''
which becomes the 
``Generalized Submodular Lower Polyhedron'' in the submodular case.}
\jeff{Rishabh}{done.}

In this section, we define a generalization of the submodular
polyhedron, which we call the \emph{generalized submodular lower
  polyhedron}. 
While this construct has not been defined explicitly before, we
investigate it primarily with the aim of contrasting this with 
results on the concave polyhedral aspects of a submodular function
that we explore in Section~\ref{polysubconcave}.

Define the generalized submodular lower polyhedron as follows:
\begin{align}
\mathcal P_f^{\text{gen}} \triangleq \{(x, c), x \in \mathbb{R}^n, c \in \mathbb{R}: [x(X) + c] \leq f(X), \forall X \subseteq V\}.
\label{eq:gen_sub_low_poly}
\end{align}
This generalized polyhedron $\mathcal P_f^{\text{gen}} \subseteq
\mathbb{R}^{n+1}$ intuitively captures the affine (or unnormalized)
modular lower bounds of $f$. The definition above holds for any
arbitrary set function, not necessarily submodular, in which case we
call it the \emph{generalized lower polyhedron}. In the case of
submodular functions, this generalized lower polyhedron has
interesting connections to the submodular polyhedron.  In particular,
note that $\mathcal P_f^{\text{gen}} \cap \{(x, c): c = 0\} = \{(x,
c): x \in \mathcal P_f, c = 0\}$. In other words, the slice $c = 0$ of
the \emph{generalized submodular polyhedron} is the submodular
polyhedron of $f$. Also notice that for a normalized submodular
function $f$, the constraint at $X = \emptyset$, requires that $c \leq
0$.

\rishabh{Jeff}{we need a figure for this. Lets talk about it
on Monday.}\jeff{Rishabh}{Maybe we can do a matlab figure plot like what you had suggested a few days back?}\jeff{Rishabh}{I've added a figure here on this.}

\begin{figure}[tbh]
\begin{center}
\includegraphics[width = 0.6\textwidth]{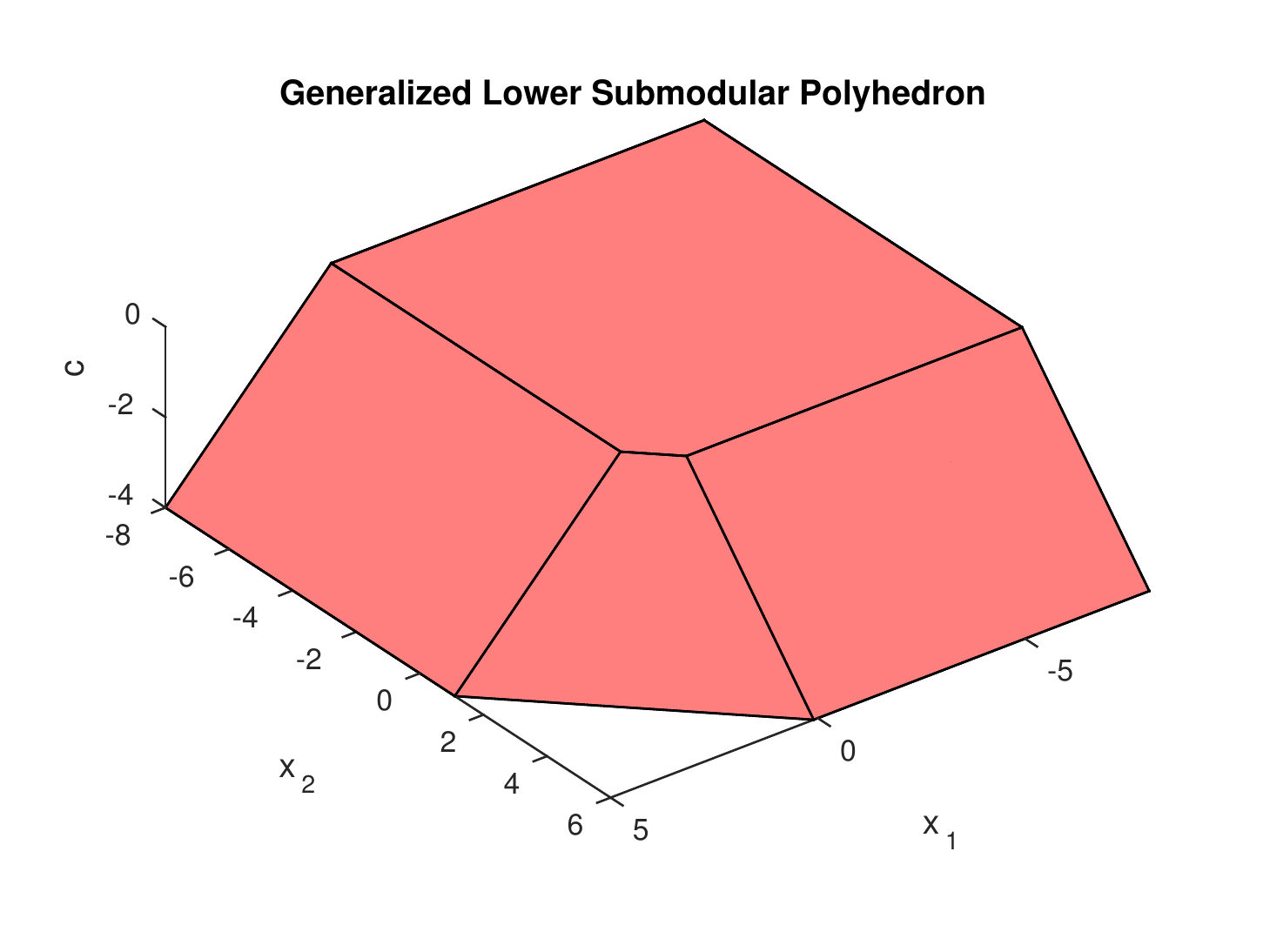}
\end{center}
\caption{The generalized submodular lower polyhedron for a two dimensional submodular function $f: 2^{\{1, 2\}} \rightarrow \mathbb{R}$, satisfying $f(\emptyset) = 0$, $f(\{1\}) = 1, f(\{2\}) = 2, f(\{1, 2\}) = 2.5$.
\rishabh{Jeff}{Can we please get a vector version of this? If this is done via matlab, there is a way to export as a .eps file that is vector, and that can be converted to pdf. This scales, prints, and converts much better than .png, not to mention it is better for slides (and more future proof). Thanks!}}\jeff{Rishabh}{I added a vector version above.}\rishabh{Jeff}{7/23/15: Great. I needed to swap the wording in
the fig to say ``submodular lower polyhedron'' and I edited the pdf directly, so if
you regenerate the fig in matlab, please change the wording.}
\label{gensubpolyvis}
\end{figure}
The generalized polyhedron has interesting connections with the subdifferential
-- the following is a characterization of its facial structure.
\begin{lemma}
Given a set function $f$, a given point $(x, c) \in \mathcal P_f^{\text{gen}}$ lies 
on a face of the polyhedron $P_f^{\text{gen}}$ if and only if there exists a set $X$ such that $x \in \partial_f(X)$ and $c = f(X) - x(X)$.
\label{lemma:face_gen_lower_poly}
 \end{lemma}  
 \begin{proof}
 Notice that $(x, c)$ lies on a face of $\mathcal P_f^{\text{gen}}$ if and only if there exists a set $X$ such that $x(X) + c = f(X)$ and for all $Y \subseteq V, x(Y) + c \leq f(Y)$. 
Since then
$x(Y) - x(X) \leq f(Y) - f(X) $, 
we have that $x \in \partial_f(X)$ and $c = f(X) - x(X)$ that, as mentioned above, has $c \leq 0$ when $f$ is submodular.
 \end{proof}
 
 The extreme points of $\mathcal P_f^{\text{gen}}$ also are easy to characterize when $f$ is submodular. Surprisingly, all the extreme points lie exactly on the hyperplane $c = 0$ with $x$ being the extreme points of $\mathcal P_f$.
 \begin{lemma}\label{gensubpolysubpolyrel}
 Given a submodular function $f$, $(x, c)$ is an extreme point of $\mathcal P_f^{\text{gen}}$ if and only if $x$ is an extreme point of $\mathcal P_f$ and $c = 0$. Furthermore, for any $y \in \mathbb{R}^n$,
\begin{align}
\underbrace{\max_{(x, c) \in \mathcal P_f^{\text{gen}}} [\langle x, y \rangle + c]}_{\text{(i)}} 
= 
\underbrace{\max\{\max_{x \in \partial_f(X)} [\langle x, y \rangle + f(X) - x(X)] \,\, | \,\ X \subseteq V\}}_{\text{(ii)}}
= 
\underbrace{\max_{x \in \mathcal P_f} \langle x, y \rangle}_{\text{(iii)}} 
\end{align}
\end{lemma}
 \begin{proof}
\rishabh{Jeff}{Also, is this true for any $y$? I think this could yield infinite values for some $y$ which doesn't violate the equation but perhaps mention it.}
\rishabh{Jeff}{I've re-done the proof below quite considerably, please take a look as there is an important question to answer still.}\jeff{Rishabh}{yes, there was a bug in the proof. I have corrected the result now. I think it should go through now.}\rishabh{Jeff}{looks good now. thanks!}

First we show that $\text{(i)} = \text{(ii)}$. Notice that a maximum
in $\text{(i)}$ (denoted by $(x^*,c^*)$) occurs at a face of
$\mathcal P_f^{\text{gen}}$ which, by
Lemma~\ref{lemma:face_gen_lower_poly}, implies that there exists an $X$ such
that $x^* \in \partial_f(X)$ with $c^*=f(X) - x(X)$. This subdifferential is considered in $\text{(ii)}$'s
outer max implying $\text{(ii)} \geq \text{(i)}$. 
Moreover, from the definitions of the generalized submodular lower
polyhedron and the subdifferential, we see it is the case that for any
$X \subset V$ and for any $x \in \partial_f(X)$, the point $(x, f(X) -
x(X)) \in \mathcal P_f^{\text{gen}}$.
Hence $\text{(ii)} \leq \text{(i)}$, since the $\max$ in (i) is over a much larger superset.

\rishabh{Jeff}{I'm not seeing why 
$\text{(ii)} \leq \text{(i)}$ since, even in 2D, for a given
direction $y$, one might hit an extreme point for one subdifferential (say
down/left closed) but for a different subdifferential (say up/right closed)
the max would be unbounded.?}\jeff{Rishabh}{I fixed the proof above. It works now.}
\rishabh{Jeff}{7/23/15: There's was still a bug, it repeats $\text{(ii)} \geq \text{(i)}$
but never shows $\text{(ii)} \leq \text{(i)}$ which is needed. I fixed it, and changed
the wording to make it clearer.}

We then show that $\text{(i)} = \text{(iii)}$.
It is immediate that 
$\text{(iii)}
\leq 
\text{(i)}$
since 
$\text{(iii)}$
is a more constrained case of 
$\text{(i)}$ under $c = 0$. 
Next, we show 
$\text{(iii)}
\geq 
\text{(i)}$, which states that for a
submodular function, the linear program over the generalized
submodular polyhedron is equivalent to a linear program over the
submodular polyhedron. This result follows as a corollary from
Lemma~\ref{greedy}. Specifically, for any
$(x, c) \in \mathcal P_f^{\text{gen}}$, we have that
\begin{align}
\max_{s \in \mathcal P_f} w^{\top} s  = \sum_i \lambda_i f(S^{\perm_w}_i) \geq \sum_i \lambda_i [\langle x, 1_{S^{\perm_w}_i} \rangle + c] \geq \langle x, w \rangle + c,
\end{align}
where the last inequality follows from the facts that $\sum_i \lambda_i 1_{S^{\perm_w}_i} = w$ and $\sum_i \lambda_i = 1$. 
In particular, this also means that in the optimization problem in $\text{ii}$, the maximum over $X \subseteq V$ occurs at $X = \emptyset$, when $\mathcal \partial_f(X) = \partial_f$.

Lastly, note that since every linear program over the generalized submodular polyhedron can be cast as a linear program over the submodular polyhedron, the extreme points of both polyhedra must also be the same.
 \end{proof}
 Intuitively, $(x, c)$ is an extreme point if $x$ is an extreme
 point of a subdifferential $\partial_f(X)$ for some set $X$. Since
 the extreme points of the subdifferentials are exactly the extreme
 points of the submodular polyhedron, the result follows.
 \rishabh{Jeff}{I don't think this adds anything to the above and
   still doesn't answer the above question.}\rishabh{Jeff}{I've changed the
text so this is ok now.}

Finally, it is worth mentioning that similar to the submodular
polyhedron, the \emph{generalized submodular polyhedron membership
  problem} (i.e., does $(x, c) \in \mathcal P_f^{\text{gen}}$) is
polynomial time, and can be solved via submodular minimization. This
is again similar to the case for the submodular (lower) polyhedron.
\begin{proposition}
Given a submodular function $f$, $(x, c) \in \mathcal P_f^{\text{gen}}$ if and only if $c \leq \min_{X \subseteq V} [f(X) - x(X)]$. 
Since submodular minimization is polynomial time, the generalized submodular polyhedral membership problem is also polynomial time.
\end{proposition}
 
A visualization of the generalized submodular lower polyhedron for a submodular function on $V = \{v_1, v_2\}$ is shown in Figure~\ref{gensubpolyvis}. 

%

\section{Convex extensions of a Submodular Function}
\label{submodcvxext}

We now describe the convex extension of a submodular functions. We
shall see a number of equivalent ways to characterize this extension
and observe how they can be computed very efficiently as what is known
as the \lovasz{} extension~\cite{edmondspolyhedra,lovasz1983}. The
results of this section are mainly taken from
\cite{edmondspolyhedra,lovasz1983, dughmi2009submodular,
  vondrak2007submodularity} and are given here both for completeness
and also to help contrast with the results we will show for the various concave
extensions given in Section~\ref{submodccvext}.

Following~\cite{vondrak2007submodularity,dughmi2009submodular}, we
consider two main characterizations of the convex extensions, as what
we call \emph{polyhedral characterization} and \emph{distributional
  characterization}. The main purpose of this section is to review
existing work thereby making it easy to contrast these results with
the new ones, on the concave extensions of submodular functions, we
shall present in Section~\ref{submodccvext},

\subsection{Polyhedral characterization of the convex extensions}
\label{sec:polyh-char-conv}

\rishabh{Jeff}{I still think we should do more broad citing of the
convex extension rather than just Dughmi which is rather recent. At
least we should cite Vondrak appropriately and I'm fairly sure
there are earlier references to the general convex extension. I think
Vondrak first proved that the convex extension is hard for non-submodular
functions, but lets do a bit of looking around. What comes
to mind is the Boros and Hammer work on pseudo-Boolean functions, they
have a book on this.}\jeff{Rishabh}{I added a reference to Vondrak and Boros and Hammer.}

The convex extension of any set function (not necessarily submodular) can be seen as the pointwise supremum of convex functions which lower bound the set function~\cite{dughmi2009submodular, vondrak2007submodularity, Boros2002155}. Precisely, let 
\begin{align}
\Phi_f \triangleq \{\phi: \phi \text{ is convex in }[0, 1]^V \text{ and } \phi(1_X) \leq f(X), \forall X \subseteq V\}.
\end{align}
be the set of continuous convex functions on $[0,1]^V$ that lower bound $f(\cdot)$.
Then define the convex extension
$\lex f : [0, 1]^{|V|} \to \mathbb R$ as follows:
\begin{align} \label{convexextcvx}
\lex f(w) \triangleq \max_{\phi \in \Phi_f} \phi(w), \text{ for } w \in [0, 1]^n
\end{align} 
It is not hard to show that $\lex f$ is convex and satisfies the relation $\lex f(1_X) = f(X)$. The above expression can in fact be simplified for any set function, and it suffices to consider affine lower instead of convex lower bounds. In particular Eqn.~\eqref{convexextcvx} can be expressed as a linear program over the generalized polyhedron.
\begin{lemma}\label{convexextafflemma}
Given a set function $f$, the convex extension of $f$ in Eqn.~\eqref{convexextcvx} can be expressed as:
\begin{align}\label{convexextaff}
\lex f(w) = \max_{(x, c) \in \mathcal P_f^{\text{gen}}} [\langle x, w \rangle + c], \forall w \in [0, 1]^n
\end{align}
\end{lemma}
\begin{proof}
  The proof of the equivalence follows from a simple observation. For
  a given $w$, let $\hat{\phi}$ be an $\argmax$ in
  Eqn.~\eqref{convexextcvx}. Then since $\hat{\phi}$ is a convex
  function in $[0, 1]^V$, there exists a subgradient $x \in
  \mathbb{R}^n$ at $w$ and value $d$, such that $\langle x, y \rangle
  + d \leq \hat{\phi}(y), \forall y$ and $\langle x, w \rangle + d =
  \hat{\phi}(w)$. In other words, $\langle x, y \rangle + d$, seen as
  a function of $y$, is a linear lower bound of $\hat{\phi}(y)$ and
  that is tight at $w$. Hence, at value $w$, $\lex f(w)$ takes value
  $\langle x, w \rangle + d$. Finally notice that $(x, d) \in \mathcal
  P_f^{\text{gen}}$ since 
  $\langle x,1_X \rangle + d = x(X) + d \leq \hat{\phi}(1_X) \leq f(X),
  \forall X \subseteq V$.
\end{proof}

In the above, we have so far not yet invoked the submodularity of $f$,
something that can lead to great simplifications. If $f$ is submodular,
then the above polyhedral characterization can be replaced by an
linear program over the submodular polyhedron. In other words,
\begin{lemma}
\label{linprogsubpoly}
For a submodular function $f$, the expressions in Eqn.~\eqref{convexextcvx}, \eqref{convexextaff} can be rewritten as:
\begin{equation}
\lex f(w) = \max_{x \in \mathcal P_f} \langle x, w \rangle, \forall w \in [0, 1]^n.
\end{equation}
\end{lemma}
\begin{proof}
This follows directly from Lemma~\ref{gensubpolysubpolyrel}.
\end{proof}
Hence, when $f$ is submodular, we may assume $c=0$ in Eqn.~\eqref{convexextaff}.
The above result is not surprising given that the extreme points of
$\mathcal P_f^{\text{gen}}$ are identical to the extreme points of
$\mathcal P_f$ when $f$ is submodular.

\subsection{Distributional characterization of the convex extension}
\label{sec:distr-char-conv}

Another way to characterize the continuous extension of a set function
$f$ is as follows. 
For a given $w \in [0, 1]^n$,
denote $\Lambda_w$ as the set:
\begin{align}
\Lambda_w \triangleq 
\Bigl\{
\{ \lambda_S, S \subseteq V \} : \sum_{S \subseteq V} \lambda_S 1_S = w, \sum_{S \subseteq V} \lambda_S = 1, \text{ and } \forall S, \lambda_S \geq 0
\Bigr\}.
\label{eq:lambda_w}
\end{align}
Then the convex extension $\lex f$ can be equivalently written as:
\begin{align}
\label{convexextdist}
\lex f(w) = \min_{\lambda \in \Lambda_w} \sum_{S \subseteq V} \lambda_S f(S)
\end{align}
The reason this representation is called distributional is that the convex extension here is computed by minimizing over particular distributions over sets. Again, it is not hard to see that this characterization is a convex extension. 

For a submodular function, the distribution characterization takes on a nice form, which is known classically as the \lovasz{} extension. 
This result can be found, for example, in~\cite{dughmi2009submodular, vondrak2007submodularity}:
\begin{lemma}~\cite{dughmi2009submodular, lovasz1983, edmondspolyhedra}\label{lovaszextlemma}
\rishabh{Jeff}{we need to add more citations to this.}\jeff{Rishabh}{Yes, here it would make sense to cite \lovasz{} and Edmonds as well. Done}
Given a submodular function $f$,
\begin{align}\label{lovaszext}
\lex f(w) = \sum_{i = 1}^n w(\perm_w(i)) (f(S^{\perm_w}_i) - f(S^{\perm_w}_{i-1}) = w(\perm_w(n))) f(S^{\perm_w}_n) + \sum_{i = 1}^{n-1} (w(\perm_w(i)) - w(\perm_w(i + 1))) f(S^{\perm_w}_i),
\end{align}
where $\sigma_w$ is a permutation satisfying $w(\sigma_w(1)) \geq w(\sigma_w(2) \geq \cdots \geq w(\sigma_w(n))$.
\end{lemma}

It is clear from above that the minimizing distribution $\lambda$ is a form of a chain distribution, where the chain here is the sequence of sets $S^{\perm_w}_0, S^{\perm_w}_1, \cdots, S^{\perm_w}_n$ defined in Lemma~\ref{greedy}. We also see the relationship between the two characterizations in the case of submodular functions, since Eqn.~\eqref{lovaszext} is exactly the solution of the linear program over the submodular polyhedron (see Lemma~\ref{greedy}). Hence the two forms of convex extensions, i.e the distributional characterization from Lemma~\ref{lovaszextlemma} and polyhedral characterization from Lemma~\ref{linprogsubpoly}, are identical for a submodular function. The resulting convex function $\lex f$ is the \lovasz{} extension.

The equivalence between the two characterizations holds for general set functions, not necessarily submodular.  In other words, Eqn.~\eqref{convexextdist} and Eqn.~\eqref{convexextcvx}, \eqref{convexextaff} are identical for any set function. This follows directly from the arguments in~\cite{dughmi2009submodular, vondrak2007submodularity}. The only catch, however, is that Lemmas~\ref{linprogsubpoly} and~\ref{lovaszextlemma} do not hold for general set functions and $\lex f$ can be NP hard to evaluate in general~\cite{dughmi2009submodular, vondrak2007submodularity, Boros2002155}.\rishabh{Jeff}{again, we need
to fix the citations here to be more appropriate and specific and ideally
cite the original sources.}\jeff{Rishabh}{True, I have added a couple more here.}

\rishabh{Jeff}{You might be wondering why I think it is important to
  cite the original sources and to broadly cite? One reason is that it
  is nice to the people who did this first. It also is nice to us
  since it shows that we are aware of the various sources where it
  exists. Thirdly, it makes for a better paper since the reader will
  learn more about the history, and it makes a paper more accessible
  since the reader will have an easier time finding background
  reading. The only reason to not cite is when we are under space
  limitations, which we are refreshingly not in the current case. :-)}\jeff{Rishabh}{Yes, I agree :-)}

\subsection{Convex Extensions and Submodular Minimization}
The \lovasz{} extension plays an important role in submodular minimization. In particular, minimizing the \lovasz{} extension is equivalent to minimizing a submodular function:
\begin{lemma} \cite{lovasz1983})
Given a submodular function $f$,
\begin{align}
\min_{X \subseteq V} f(X) = \min_{x \in [0, 1]^n} \lex f(x)
\end{align}
Furthermore, given the minimizer $x^*$ of the RHS above, we can obtain a set $X^*$ such that $f(X^*) = \lex f(x^*)$.
\end{lemma}
This implies that unconstrained submodular minimization has an integrality gap of one and the two problems are equivalent.

\section{Optimality conditions for submodular minimization}
\label{submodminopt}

Fujishige~\cite{fujishige1984theory} provides some interesting
characterizations of optimality conditions for unconstrained
submodular minimization. The following theorem can be thought of as a
discrete analog to the KKT conditions:
\begin{lemma}(\cite[Lemma 7.1]{fujishige2005submodular})\label{optsubmin}
A set $A \subseteq V$ is a minimizer of $f: 2^V \rightarrow \mathbb{R}$ if and only if:
\begin{align}
\textbf{0} \in \partial_f(A)
\end{align}
\end{lemma}
This immediately provides necessary and sufficient conditions for optimality of $f$:
\begin{lemma}(\cite[Theorem 7.2]{fujishige2005submodular})
A set $A$ minimizes a submodular function $f$ if and only if $f(A) \leq f(B)$ for all sets $B$ such that $B \subseteq A$ or $A \subseteq B$.
\end{lemma}
In other words, it is sufficient to check only the subsets and supersets of $A$ to ensure that $A$ is a global optimizer of $f$. The above Lemma follows from Eqn.~\eqref{subdiffred} and Lemma~\ref{optsubmin}. Analogous characterizations have also been provided for constrained forms of submodular minimization, and interested readers may look at~\cite{fujishige1984theory}. Finally, we can provide a simple characterization on the local minimizers of a submodular function.
\begin{lemma}
A set $A \subseteq V$ is a local minimizer\footnote{A set $A$ is a local minimizer of a submodular function if $f(X) \geq f(A), \forall X: |X \backslash A| \leq 1, \text{ and } |A \backslash X| = 1$, that is all sets $X$ no more than hamming-distance one away from $A$.} of a submodular function if and only if $\textbf{0} \in \partial_f^{\symmdiff(1, 1)}(A)$. 
\end{lemma}
As was shown in \cite{rkiyersemiframework2013}, a local minimizer of a submodular function, in the unconstrained setting, can be found efficiently in $O(n^2)$ complexity. 

While unconstrained submodular minimization is easy, most forms of constrained submodular minimization become NP hard. For example, a simple cardinality lower bound constraint makes the problem of submodular minimization (even with monotone submodular functions) NP hard without even constant factor approximation guarantees~\cite{svitkina2008submodular}. These results, however, can be extended when the constraints are lattice constraints~\cite{fujishige2005submodular} in which case many of the results above still hold.

\section{Convex Characterizations: Discrete Separation Theorem and Fenchel Duality Theorem}
\label{submodfdtdstmin}

We next review some interesting theorems that characterize convex
functions, and that interestingly also hold for submodular functions.

\subsection{The Discrete Separation Theorem (DST)}
\label{sec:discr-seper-theor}

The separation theorem \cite{rockafellar1970convex}, known in context
of convexity, states that given a convex function $\phi$ and a concave
function $\psi$ such that $\forall x, \phi(x) \geq \psi(x)$, there
exists an affine function $\langle h, x \rangle + c$ such that $\forall
x, \psi(x) \geq \langle h, x \rangle + c \geq \psi(x)$. 

A similar relation holds for submodular functions. The lemma below was
shown by Frank~\cite{frank1982algorithm} and has become known
as the {\em discrete separation theorem} (DST):
\begin{lemma}~\cite{frank1982algorithm}, \cite[Theorem 4.12]{fujishige2005submodular}
Given a submodular function $f$ and a supermodular function $g$ such that $f(X) \geq g(X), \forall X$ (and which satisfy $f(\emptyset) = g(\emptyset) = 0$), there exists a modular function $h$ such that $f(X) \geq h(X) \geq g(X)$. Furthermore, if $f$ and $g$ are integral so may be $h$.
\label{lemma:orig_dst}
\end{lemma}
This Lemma can also be shown using the \lovasz{} extension. In
particular, given a submodular function $f$ and a supermodular
function $g$ such that $f(X) \geq g(X), \forall X$, we can construct
the convex and concave extensions $\lex f$ and $\cex g$ of $f$ and $g$
(the concave extension $\cex {g}$ can be constructed via the \lovasz{}
extension of $-g$). From the expressions of $\lex f$ and $\cex g$, it
is not hard to see that $\lex f(x) \geq \cex g(x), \forall x$. Hence
using the separation theorem from convex analysis, we can
find a linear function $h$, which when restricted to 0/1 vectors,
gives the modular function $h$.

The DST is one of the results that shows how submodular functions are
analogous to convex functions. Surprisingly, we will show in
Section~\ref{sec:discr-seper-theor-concave} that a form of opposite
(and slightly restricted) DST also holds for submodular functions that
relates submodularity to concave functions.

\subsection{Fenchel Duality Theorem (FDT)}  
\label{sec:fench-dual-theor}

The Fenchel duality theorem in the context of convexity \cite{rockafellar1970convex}
provides a relation between the minimizers of the function and it is dual. Given a convex function $\phi$ and a concave function $\psi$, the Fenchel dual $\phi^*$ of
$\phi$, and $\psi^*$ of $\psi$, is given as follows:
\begin{align}
\phi^*(y) = \max_{x \in \text{dom}(\phi)} [\langle x, y \rangle - \phi(x)]
\qquad 
\text{ and }
\qquad 
\psi^*(y) = \min_{x \in \text{dom}(\psi)} [\langle x, y \rangle - \psi(y)].
\end{align}
The dual functions $\phi^*$ and $\psi^*$ are convex and concave respectively. The Fenchel duality theorem then states that:
\begin{align}
\min_x [\phi(x) - \psi(x)] = \max_y [\psi^*(y) - \phi^*(y)]
\end{align}

Analogous characterizations also hold for submodular functions~\cite{fujishige1984theory}. Given a submodular function $f$ (or equivalently supermodular function $g$), 
the Fenchel dual $f^*$ of $f$,
and $g^*$ of $g$, are defined as follows:
\begin{align}
f^*(x) = \max_{X \subseteq V} [x(X) - f(X)]
\qquad
\text{ and }
\qquad
g^*(x) = \min_{X \subseteq V} [x(X) - g(X)].
\end{align}
The Fenchel duals $f^*$ and $g^*$ are convex and concave functions respectively. 
Then the following Lemma for submodular functions, analogous to
the case for convex and concave functions, holds:
\begin{lemma}~(\cite[Theorem 6.3]{fujishige2005submodular})
Given a submodular function $f$ and a supermodular function $g$,
\begin{align}
\min_{X \subseteq V} [f(X) - g(X)] = \max_x [g^*(x) - f^*(x)].
\end{align}
Further if $f$ and $g$ are integral, the maximum on the right hand side is attained by an integral vector $x$.
\end{lemma}

\subsection{The Minkowski sum theorem} 
\label{sec:mink-sum-theor}

Submodular polyhedra and also the subdifferentials have an interesting
characterization related to Minkowski sums of polyhedra, namely $P+Q
\triangleq \{x+y : x \in P \text{ and } y \in Q \}$ for polyhedra $P$
and $Q$.
\begin{lemma}(\cite[Theorem 6.8]{fujishige2005submodular})
Given two submodular functions $f_1$ and $f_2$, it holds that the addition of the polyhedra corresponds to a point-wise addition. That is:
\begin{align}
\mathcal P_{f_1 + f_2} = \mathcal P_{f_1} + \mathcal P_{f_2},
\qquad
\text{ and more generally, }
\quad \partial_{f_1 + f_2}(X) = \partial_{f_1}(X) + \partial_{f_2}(X)
\end{align}
Similarly it holds that $\mathcal P^{\text{gen}}_{f_1 + f_2} = \mathcal P^{\text{gen}}_{f_1} + \mathcal P^{\text{gen}}_{f_2}$.

\end{lemma}
The Minkowski sum theorem for the generalized submodular polyhedron
follows directly from the definition.

\section{Concave Polyhedral Aspects of Submodular Functions}
\label{polysubconcave}

We next investigate several polyhedral aspects of submodular functions
relating them to concavity, thus complementing the results from
Section~\ref{polysubconvex}. This provides a complete picture on the
relationship between submodularity, convexity, and concavity. We
define and investigate the submodular upper polyhedron, submodular
superdifferential, and the generalized submodular upper polyhedron.

\subsection{The submodular upper polyhedron}
\label{uppersubpolysec}

\begin{figure}
\centering
\includegraphics[page=1,width = 0.2\textwidth]{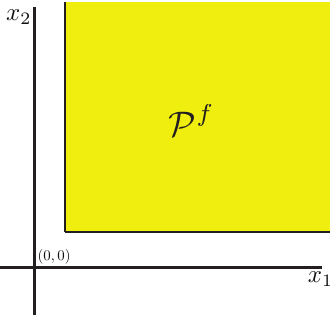}
\caption{The submodular upper Polyhedron $\mathcal P^f$ in two dimensions.}
\end{figure}

A first step in characterizing the concave aspects of a submodular function is the submodular upper polyhedron. Intuitively this is the set of tight modular upper bounds of the function,
and we define it as follows:
\begin{align}
\mathcal P^f \triangleq \{x \in \mathbb{R}^n: x(S) \geq f(S), \forall S \subseteq V\}
\label{eqn:submodular_upper_polyhedron}
\end{align}
The above polyhedron can in fact be defined for any set function. In
particular, when $f$ is supermodular, we get what is known as the
\emph{supermodular polyhedron} \cite{fujishige2005submodular}.
Presently, we are interested in the case when $f$ is submodular and
hence we call this the \emph{submodular upper polyhedron}, a construct
that is quite different than the supermodular polyhedron.

Interestingly, submodular upper polyhedron has a very simple
characterization due to the submodularity of $f$. We have
the following:
\begin{lemma}\label{antisubpoly}
Given a submodular function $f$, 
\begin{align}
\mathcal P^f = \{x \in \mathbb{R}^n: x(j) \geq f(j)\}
\label{eqn:submodular_upper_polyhedron_n_inequalities}
\end{align}
\end{lemma}
\begin{proof}
Given $x \in \mathcal P^f$ and a set $S$, 
we have 
$x(S) = \sum_{i \in S} x(i) \geq \sum_{i \in S} f(i)$, since $\forall i, x(i) \geq f(i)$
by Eqn.~\eqref{eqn:submodular_upper_polyhedron}.
Hence $x(S) \geq \sum_{i \in S} f(i) \geq f(S)$. Thus, the irredundant inequalities 
are the singletons.
\end{proof}
The lemma states that this polyhedron is not \emph{polyhedrally tight}
in that the vast majority of the defining inequalities are redundant.
Unlike the submodular lower polyhedron, the submodular upper
polyhedron is not particularly interesting or useful for defining a
concave extension. We shall, however, define a generalization of the submodular
upper polyhedron in Section~\ref{sec:gener-subm-upper} that
will prove quite useful in characterizing, and providing
approximations to, the concave extension of $f$.

We end this section by investigating the submodular upper polyhedron
membership problem. Owing to its simplicity, this problem is
particularly simple which might seem surprising at first glance since,
from Eqn.~\eqref{eqn:submodular_upper_polyhedron}, the problem of
checking $x \in \mathcal P^f$ is equivalent to checking if $\max_{X
  \subseteq V} f(X) - x(X) \leq 0$. In general, this would involve the maximization of a
submodular function which is NP hard.  The following lemma shows that
this particular problem is actually easy.
\begin{corollary}
Given a submodular function $f$ and vector $x$, let $X$ be a set such that $f(X) - x(X) > 0$. Then there exists an $i \in X: f(i) - x(i) > 0$.
\end{corollary}
\begin{proof}
  Observe that $f(X) - x(X) \leq \sum_{i \in X} f(i) - x(i)$. Since
  the l.h.s.\ is greater than $0$, it implies that $\sum_{i \in X} f(i)
  - x(i) > 0$. Hence there should exist an $i \in X$ such that $f(i) -
  x(i) > 0$.
\end{proof}
Thus, it is sufficient to check the singleton values, i.e $f(i) -
x(i)$, and if all these are less than or equal to zero, then $x \in
\mathcal P^f$. This also follows immediately from
Lemma~\ref{antisubpoly}.

An interesting corollary of the above is that it is easy to
check if the maximizer of a submodular function is greater than or equal
to zero. Given a submodular function $f$, the problem is whether
$\max_{X \subseteq V} f(X) \geq 0$. This can easily be checked without
resorting to submodular function maximization.
\begin{corollary}
Given a submodular function $f$ with $f(\emptyset) = 0$, $\max_{X \subseteq V} f(X) > 0$ if and only if there exists an $i \in V$ such that $f(i) > 0$.
\end{corollary}
\begin{proof}
If for any $j$, $f(j) > 0$ it implies that $\max_{X \subseteq V} f(X) \geq f(j) > 0$. On the other hand, if $\forall j, f(j) \leq 0$, we have that $\forall X \subseteq V, f(X) \leq \sum_{i \in X} f(i) \leq 0$. Hence $\max_{X \subseteq V} f(X) = 0$.
\end{proof}

This fact is true only for a submodular function. For general set
functions, even when $f(\emptyset) = 0$, it could potentially require
an exponential cost search to determine if $\max_{X \subseteq V} f(X)
> 0$.

\subsection{The Submodular Superdifferentials}
\label{sec:subm-superd}

Given a submodular function $f$, we can characterize its
superdifferentials that constitute a partition of $\mathbb R^n$.  Given
any $X \subseteq V$, we denote superdifferential with respect to $X$
as $\partial^f(X)$ and define it as follows:
\begin{equation}
\label{supdiff-def}
\partial^f(X) \triangleq \{x \in \mathbb{R}^n: f(Y) - x(Y) \leq f(X) - x(X), \forall Y \subseteq V \}
\end{equation}
This characterization is analogous to the subdifferential of a
submodular function defined in Eqn.~\eqref{submodsubdiff}. This is
also akin to the superdifferential corresponding to a continuous
concave function. Since, as we will see, submodular functions have
both a subdifferential and a superdifferential structure that are
distinct, the two names will, correspondingly, refer to distinct
constructs.

Each supergradient $g_X \in \partial^f(X)$ defines a modular upper
bound of a submodular function. In particular, define the following
modular function:
\begin{align}
m^X(Y) \triangleq f(X) + g_X(Y) - g_X(X).
\label{eq:mod_tight_upper_bound}
\end{align}
Then $m^X(Y)$ is a modular function which satisfies
$m^X(Y) \geq f(Y), \forall Y \subseteq X$ and $m^X(X) = f(X)$. This is
analogous to the submodular subdifferential in
Section~\ref{sec:subm-subd} 
(and in particular Eqn.~\eqref{eq:mod_tight_lower_bound})
where tight modular lower bounds were
produced --- here, we produce tight modular upper bounds on any
submodular function.

We note that $(x(v_1), x(v_2), \dots, x(v_n)) =
(f(v_1),f(v_2),\dots,f(v_n)) \in
\partial^f(\emptyset)$ which shows at least that $\partial^f(\emptyset)$ exists.
A bit further below (specifically Theorem~\ref{altviewsthm2}) we show that for any submodular function, $\partial^f(X)$ is non-empty
for all $X \subseteq V$.
\rishabh{Jeff}{I think that this would be a good place to offer a vector that exists in $\partial^f(X)$ for all $X$, for example $(f(v_1),f(v_2),\dots,f(v_n)) \in 
  \partial^f(X)$
  for some $X$ which shows that there are $X$ for which it is
  empty. While this is not an interesting supergradient, it at least
  offers trivial existence right away, perhaps this could be done for
  all $X$.}\rishabh{Jeff}{I've done this now in the paragraph above.}\jeff{Rishabh}{thanks Jeff! Yes, this gives a better clarity early on to the reader.}

Note that the superdifferential is defined by an exponential (i.e.,
$2^{|V|}$) number of inequalities. However owing to the submodularity
of $f$ and akin to the subdifferential of $f$, we can reduce the
number of inequalities since some of them are redundant given the others. Define three
polyhedrons as follows:
\begin{align}
\partial^f_1(X) &\triangleq \{x \in \mathbb{R}^n: f(Y) - x(Y) \leq f(X) - x(X), \forall Y \subseteq X \}, 
\label{eq:superdifferential_first_bit} \\
\partial^f_2(X) &\triangleq \{x \in \mathbb{R}^n: f(Y) - x(Y) \leq f(X) - x(X), \forall Y \supseteq X \}, 
\label{eq:superdifferential_second_bit} \\
\partial^f_3(X) &\triangleq \{x \in \mathbb{R}^n: f(Y) - x(Y) \leq f(X) - x(X), \forall Y: Y \not \subseteq X, Y \not \supseteq X\}. \label{eq:superdifferential_third_bit}
\end{align}
An immediate observation is that 
\begin{align}
\partial^f(X) = \partial^f_1(X) \cap \partial^f_2(X) \cap \partial^f_3(X).
\label{eq:super_intersection}
\end{align}
As we show below for a submodular function $f$, $\partial^f_1(X)$ and $\partial^f_2(X)$ are actually very simple polyhedra.
\begin{lemma}
For a submodular function $f$, 
\begin{align}
\partial^f_1(X) &= \{x \in \mathbb{R}^n: f(j | X \backslash j) \geq x(j), \forall j \in X\} \\
\partial^f_2(X) &= \{x \in \mathbb{R}^n: f(j | X) \leq x(j), \forall j \notin X\}.
\end{align}
\label{thm:super_partial_one_two_is_esy}
\end{lemma}
\begin{proof}
Consider $\partial^f_1(X)$. Notice that the inequalities defining the polyhedron,
starting from Eqn.~\eqref{eq:superdifferential_first_bit},
can be rewritten as $\partial^f_1(X) = \{x \in \mathbb{R}^n:x(X \backslash Y) \leq f(X) - f(Y), \forall Y \subseteq X \}$. We then have that $x(X \backslash Y) = \sum_{j \in X \backslash Y} x(j) \leq \sum_{j \in X \backslash Y} f(j | X \backslash j)$, since $\forall j \in X, x(j) \leq f(j | X \backslash j)$ (this follows by considering only the 
subset of inequalities of $\partial^f_1(X)$ 
in Eqn.~\eqref{eq:superdifferential_first_bit}
with sets $Y \subseteq X$ such that $|X \backslash Y| = 1$). We then have $x(X \backslash Y) \leq \sum_{j \in X \backslash Y} f(j | X \backslash j) \leq f(X) - f(Y)$ where the last inequality follows from submodularity alone. Hence 
an irredundant set of inequalities include those defined only through the singletons. 


 In order to show the characterization for $\partial^f_2(X)$, we have,
starting from Eqn.~\eqref{eq:superdifferential_second_bit},
that $\partial^f_2(X) = \{x \in \mathbb{R}^n:x(Y \backslash X) \geq f(Y) - f(X), \forall Y \supseteq X \}$. It then follows that, $x(Y \backslash X) = \sum_{j \in Y \backslash X} x(j) \geq \sum_{j \in Y \backslash X} f(j | X)$, since $\forall j \notin X, x(j) \geq f(j | X)$. Hence $x(X \backslash Y) \geq \sum_{j \in Y \backslash X} f(j | X) \geq f(Y) - f(X)$, and again, an irredundant set of inequalities include those defined only through the singletons. 
\end{proof}

The above characterization shows that $\partial^f(X)$ can be determined using many fewer inequalities since the polytopes $\partial^f_1(X)$ and $\partial^f_2(X)$ are so simple. Recall that this is analogous to the submodular subdifferential, where again owing to submodularity the number of essential inequalities can be reduced significantly --- in that case, we just need to consider the sets $Y$ which are subsets and supersets of $X$. It is interesting to note the contrast between the redundancy of inequalities in the subdifferentials and the superdifferentials. In particular, here, the inequalities corresponding to sets $Y$ being the subsets and supersets of $X$ are mostly redundant, while the irredundant ones are the rest of the inequalities. In other words, in the case of the subdifferential, $\partial_f^1(X)$ and $\partial_f^2(X)$ were essential, while $\partial_f^3(X)$ was entirely redundant given the first two.  In the case of the superdifferentials, $\partial^f_1(X)$ and $\partial^f_3(X)$ are mostly internally redundant (they can be represented using only by $n$ inequalities), while $\partial^f_3(X)$ has no redundancy in general.

In order to gain more intuition for the superdifferentials, we next consider
some examples in both two and three dimensions.


\begin{figure}[tbh]
\centering
\includegraphics[page=3,width=0.5\textwidth]{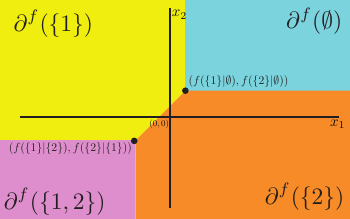}
\caption{A visualization of the four submodular superdifferentials $\partial^f(Y)$ for different sets $Y$ in two dimensions $V=\{v_1, v_2\}$, as described in Example~\ref{ex:2dsuperdifferential}.
}
\label{fig:2d_superdifferential_figure}
\end{figure}

\begin{example}
\label{ex1}
See Figure~\ref{fig:2d_superdifferential_figure}.
We consider here superdifferentials of a submodular function when $V = \{1, 2\}$. Then from the lemma above, $\partial^f(\emptyset) = \{x \in \mathbb{R}^2: f(j | \emptyset) \leq x(j), \forall j \in \{1, 2\}\}$. Similarly $\partial^f(\{1, 2\}) = \{x \in \mathbb{R}^2: f(j | V \backslash j) \geq x(j), \forall j \in \{1, 2\}\}$. Now consider $\partial^f(1)$. Then the governing inequalities for this are:
\begin{align}
\partial^f(1) = \{x \in \mathbb{R}^2: &x_1 \leq f(\{1\}), \label{eq1d2}\\
			 							& x_2 \geq f(\{2\} | \{1\}), \label{eq2d2}\\
										& x_1 - x_2 \leq f(\{1\}) - f(\{2\}) \label{eq3d2}
\end{align}
The extreme points of this polyhedron are the vectors $\{f(\{1\}),
f(\{2\})\} = \{f(\{1\}| \emptyset), f(\{2\}| \emptyset)\} $ and
$\{f(\{1\} | \{2\}), f(\{2\} | \{1\})\}$. The way we obtain the
extreme points is as follows. Setting the inequalities \eqref{eq1d2}
and \eqref{eq3d2} as equalities, we get the extreme point $\{f(\{1\}),
f(\{2\})\}$. The inequality \eqref{eq2d2} then is $x_2 = f(\{2\}) \geq
f(\{2\} | \{1\})$, which holds. We then set inequalities \eqref{eq2d2}
and \eqref{eq3d2} as equalities, which gives $x_2 = f(\{2\} | \{1\})$
and $x_1 = f(\{1\} | \{2\})$, thus giving the second extreme point
$\{f(\{1\} | \{2\}), f(\{2\} | \{1\})\}$. One can see that
inequality \eqref{eq1d2} is satisfied. Finally, if we set inequalities
\eqref{eq1d2} and \eqref{eq2d2} as equalities, we get $x_1 = f(\{1\}),
x_2 = f(\{2\} | \{1\})$. Then we have that $x_1 - x_2 = f(\{1\}) -
f(\{2\} | \{1\}) = 2f(\{1\}) - f(\{1, 2\})$. Inequality \eqref{eq3d2},
then requires, $2f(\{1\}) - f(\{1, 2\}) \leq f(\{1\}) - f(\{2\})
\Rightarrow f(\{1\}) + f( \{2\}) \leq f(\{1, 2\})$, which does not
hold (unless $f$ is trivially modular, in which case the first two
extreme points collapse onto the third). Hence the only extreme points
are the two vectors above. One can similarly investigate
$\partial^f(\{2\})$, which has the same extreme points.
\label{ex:2dsuperdifferential}
\end{example}

It is clear from the above example that superdifferentials in the
two-dimensional case are easy to find and characterize. However this
is not the case in three dimensions where the shape of the
superdifferentials depends strongly on the particulars of the
submodular functions --- this means that one cannot characterize the
superdifferential polyhedra knowing only that $f$ is submodular, more
information about the specific instance is required.


\begin{example}\label{ex2}
Let $V = \{1,2,3\}$. Recall that $f(\emptyset) = 0$. Then consider 
$\partial^f(\{1\}) 
= \{x \in \mathbb{R}^3: f(Y) - x(Y) \leq f(\{1\}) - x(\{1\}), \forall Y \subseteq \{1,2,3\} \}$
This polyhedron can be represented via the following irredundant inequalities:
\begin{alignat}{2}
 \partial^f(\{1\}) = 
\{x \in \mathbb{R}^3:  
     x_1 &\leq f(\{1\}),\quad\quad\quad           &\text{ for } Y&=\{\emptyset\}   \label{eq1}\\ 
     x_2 &\geq f(\{2\}|\{1\}),\quad\quad              &\text{ for } Y&=\{1,2\}   \label{eq2}\\
     x_2 - x_1 &\geq f(\{2\}) - f(\{1\}),\quad\quad  &\text{ for } Y&=\{2\}   \label{eq3}\\
     x_3 &\geq f(\{3\}|\{1\}),                        &\text{ for } Y&=\{1,3\}               \label{eq4}\\
     x_3 - x_1 &\geq f(\{3\}) - f(\{1\}),\quad\quad      &\text{ for }Y&=\{3\}               \label{eq5}\\
     x_2 + x_3 - x_1 &\geq  f(\{2,3\}) - f(\{1\})  &\text{ for }Y&=\{2,3\}   
\label{eq:3dsuper}\} 
\end{alignat}
The other two inequalities (for $Y=\{1,2,3\}$ and $Y=\{1\}$) are redundant given the above.  We now
consider the extreme points of this polyhedron. We consider
Eqns.~\eqref{eq1}, \eqref{eq3}, \eqref{eq5} with equality, and
we obtain an extreme point $\{f(\{1\}), f(\{2\}), f(\{3\})\}$. It is the case that all other inequalities are satisfied.  Consider
next Eqns.~\eqref{eq3}, \eqref{eq5}, and \eqref{eq:3dsuper} with
equality and we get a potential extreme point $\{f(\{1\}) + f(\{2,3\}) - f(\{2\}) -
f(\{3\}), f(\{2\}|\{3\}), f(\{3\}|\{2\})\}$. Observe that $x_1 = f(\{1\}) + f(\{2, 3\}) - f(\{2\})
- f(\{3\}) \leq f(\{1\})$, and hence Eqn.~\eqref{eq1} is satisfied. However
$x_2 = f(\{2\}|\{3\})$ may be bigger or smaller than $f(\{2\}|\{1\})$ (depending on the specific submodular function instance) and
Eqn.~\eqref{eq2} might or might not be violated. Similarly, $x_3 = f(\{3\} | \{2\})$ is not comparable to
$f(\{3\} | \{1\})$, and hence Eqn.~\eqref{eq4} might or might not be
violated. Consequently we cannot
determine if by combining together Eqns.~\eqref{eq3}, \eqref{eq5}, and \eqref{eq:3dsuper}, we obtain an extreme point, unless we have more information
about the current submodular function being used. We therefore, from this example, see that we cannot hope to find
the extreme points both analytically and generically.
\label{ex:3dsuperdifferential}
\end{example}

The above example shows that a particular expression (obtained via a combination of inequalities) might or might not be extreme, depending on the particular submodular function and its valuation. This is unlike the subdifferential, where a certain analytical expression is always extreme for all submodular functions.
Thus, unlike the subdifferentials, we cannot expect a closed form
expression for the extreme points of $\partial^f(X)$. Moreover, they
also seem to be hard to characterize algorithmically. For example, the
superdifferential membership problem is NP hard.


\begin{lemma}\label{NPsuperdiffmembership}
Given a submodular function $f$ and a set $X: \emptyset \subset X \subset V$, the membership problem $y \in \partial^f(X)$ is NP hard.
\end{lemma}
\begin{proof}
Notice that the membership problem $y \in \partial^f(X)$ is equivalent to asking $\max_{Y \subseteq V} f(Y) - y(Y) \leq f(X) - y(X)$. In other words, this is equivalent to asking if $X$ is a maximizer of $f(Y) - y(Y)$ for a given vector $y$. This is the decision version of the submodular maximization problem and correspondingly is NP hard when $\emptyset \subset X \subset V$. 
\end{proof}

Given that the membership problem is NP hard, it is also NP hard to solve a linear program over this polyhedron~\cite{grotschel1984geometric, schrijver2003combinatorial}.
The superdifferential for the empty set $X=\emptyset$ and the ground set
$X=V$, however, can be characterized easily:
\begin{lemma}
For any submodular function $f$ such that $f(\emptyset) = 0$, $\partial^f(\emptyset) = \{x \in \mathbb{R}^n: f(j) \leq x(j), \forall j \in V\}$. Similarly $\partial^f(V) = \{x \in \mathbb{R}^n: f(j | V \backslash j) \geq x(j), \forall j \in V\}$. Furthermore, $\partial^f(\emptyset) = \mathcal P^f$ and $\partial^f(V) = \mathcal P^{f^{\#}}$.
\label{thm:super_at_empty_and_full}
\end{lemma}
\begin{proof}
Consider $X = \emptyset$, 
then $\partial^f(\emptyset) = \{ x \in \mathbb R^n : f(Y) \leq x(Y), \forall Y \subseteq V \}$. 
Assuming only the $|V|$ inequalities for $Y = \{ j \} \in V$ gives
$f(Y) \leq \sum_{j \in Y} f(j) \leq \sum_{j \in Y} x(j) = x(Y)$ meaning
only the these $|V|$ inequalities are necessary.
For $X = V$, 
$\partial^f(V) = \{ x \in \mathbb R^n : x(V) - x(Y) \leq f(V) - f(Y), \forall Y \subseteq V \}$. 
Assuming only the $|V|$ inequalities for $Y = V \setminus \{ j \}$ gives
$\sum_{j \in V \setminus Y} x(j) 
\leq
\sum_{j \in V \setminus Y} f(j | V \setminus j )
\leq f(V\setminus Y | Y) = f(V) - f(Y)$.
This rest follows directly from the definitions.
\end{proof}
This also follows
by first noting that in
Eqn.~\eqref{eq:superdifferential_third_bit} we have $\partial^f_3(\emptyset)
= \partial^f_3(V) = \mathbb R^n$. Then,
by using 
Eqn~\eqref{eq:super_intersection} and
Lemma~\ref{thm:super_partial_one_two_is_esy}, 
we have that 
Eqn.~\eqref{eq:superdifferential_first_bit} implies $\partial^f_1(\emptyset) =
\mathbb R^n$
so that $\partial^f(\emptyset) = \partial^f_2(\emptyset)$, 
and that Eqn.~\eqref{eq:superdifferential_second_bit} implies
$\partial^f_2(V) = \mathbb R^n$
so that $\partial^f(V) = \partial^f_1(V)$.

As we see from the above, it is hard to characterize the superdifferential of a
submodular function. It is however possible to provide computationally
feasible inner and outer bounds as shown in the following
subsections. Using these, we can also find certain specific and practically useful
supergradients.

\subsubsection{Outer bounds on the superdifferential}
\label{sec:outer-bounds-superd}


It is possible to provide a number of useful and practical outer bounds on the superdifferential. 
Recall from Lemma~\ref{thm:super_partial_one_two_is_esy}
that $\partial^f_1(X)$ and $\partial^f_2(X)$,
defined in Eqns.\eqref{eq:superdifferential_first_bit}
and~\eqref{eq:superdifferential_second_bit},
are already simple polyhedra. We can then provide outer bounds on $\partial^f_3(X)$ that,
together with $\partial^f_1(X)$ and $\partial^f_2(X)$,
provide simple bounds on $\partial^f(X)$. Define for $1 \leq k,l \leq n$:
\begin{align}
\partial^f_{3, \symmdiff(k, l)}(X) &\triangleq \{x \in \mathbb{R}^n: f(Y) - x(Y) \leq f(X) - x(X), \nonumber \\ 
& \forall Y: Y \not \subseteq X, Y \not \supseteq X, |Y \backslash X| \leq k-1, |X \backslash Y| \leq l-1\}
\end{align}
Note that for $\emptyset \subseteq X \subseteq V$, 
$\partial^f_{3, \symmdiff(n, n)}(X) = \partial^f_{3}(X)$
and $\partial^f_{3, \symmdiff(k, l)}(X) \supseteq \partial^f_{3}(X)$ for $1 \leq k,l \leq 
n$.
We also have that $\partial^f_{3, \symmdiff(1, 1)}(X) = \mathbb R^n$.
We can then define the outer bound: 
\begin{align}
\partial^f_{\symmdiff(k, l)}(X) \triangleq \partial^f_1(X) \cap \partial^f_2(X) \cap \partial^f_{3, \symmdiff(k, l)}(X) \supseteq \partial^f(X). 
\label{eq:super_diff_outer_bounds}
\end{align}
Observe that $\partial^f_{\symmdiff(k, l)}(X)$ is expressed in terms of $O(n^{k+l})$ inequalities, and hence for a given constant $k, l$ we can obtain the representation of $\partial^f_{\symmdiff(k, l)}(X)$ in polynomial time. We will see that this provides us with a hierarchy of outer bounds on the superdifferential:
\begin{theorem}
For a submodular function $f$:
\begin{enumerate}
\item $\partial^f_{\symmdiff(1, 1)}(X) = \partial^f_1(X) \cap \partial^f_2(X)$
\item $\forall 1 \leq k^{\prime} \leq k , 1 \leq l^{\prime} \leq l, \partial^f(X) \subseteq \partial^f_{\symmdiff(k, l)}(X) \subseteq \partial^f_{\symmdiff(k^{\prime}, l^{\prime})}(X) \subseteq \partial^f_{\symmdiff(1, 1)}(X)$
\item $\partial^f_{\symmdiff(n, n)}(X) = \partial^f(X)$.
\end{enumerate}
\end{theorem}
\begin{proof}
The proofs of items 1 and 3 follow directly from definitions. To see item 2, notice that the polyhedra $\partial^f_{\symmdiff(k, l)}$ become tighter as $k$ and $l$ increase finally approaching the superdifferential.
\end{proof}
Similar to how Eqn.~\eqref{localsub} relates to the submodular
subdifferential, we shall call $\partial^f_{\symmdiff(1, 1)}(Y)$
the local approximation of the superdifferential. In particular,
\begin{align}
\label{eq:localsup_long}
\partial^f_{\symmdiff(1, 1)}(X) 
&= \{ x \in \mathbb R^n : f(Y) - x(Y) \leq f(X) - x(X), \forall Y \in [ \emptyset, X ] \cup [X , V ] \} 
\\
&= \{x \in \mathbb{R}^n: \forall j \in X, f(j | X \backslash j) \geq x(j) \text{ and } 
\forall j \notin X, f(j | X) \leq x(j) \}
\label{localsup}
\end{align}
where the second equality follows from Lemma~\ref{thm:super_partial_one_two_is_esy}.
It is interesting to note that the very same irredundant sets
that equivalently define the subdifferential in Eqn.~\eqref{subdiffred} are
also the ones that define an outer bound of the superdifferential in
Eqn.~\eqref{eq:localsup_long}.

We shall see in Section~\ref{submodmaxopt} that these outer bounds have interesting connections with approximation algorithms for submodular maximization.

\subsubsection{Inner Bounds on the superdifferential}
\label{sec:inner-bounds-superd}

While it is hard to characterize the extreme points of the
superdifferential, we can provide some specific and useful
supergradients.  
For any $X \subseteq V$,
define three vectors 
$\ggrow_X,
\gshrink_X,
\gbar_X \in \mathbb R^n$
as follows:\looseness-1
\begin{align}
\ggrow_X(j) \triangleq 
\begin{cases}
f(j |X \setminus j) & \text{ if }  j \in X\\
f(j) & \text { if } j \notin X \\
\end{cases},
\label{eq:ggrow_def} \\
\gshrink_X(j) \triangleq 
\begin{cases}
f(j |V \setminus j) & \text{ if }  j \in X\\
f(j|X) & \text { if } j \notin X
\end{cases},
\label{eq:gshrink_def} \\
\text{ and } \qquad \gbar_X(j) \triangleq 
\begin{cases}
f(j |V \setminus j) & \text{ if }  j \in X\\
f(j) & \text { if } j \notin X
\end{cases}.
\label{eq:gbar_def}
\end{align}

Then we have the following theorem:
\begin{theorem} \label{altviewsthm2} 
For a submodular function $f$, $\ggrow_X, \gshrink_X, \gbar_X \in \partial^f(X)$. Hence for every submodular function $f$ and set $X$, $\partial^f(X)$ is non-empty.
\end{theorem}
\begin{proof}
For submodular $f$, the following bounds are known to hold~\cite{nemhauser1978}
for all $X,Y \subseteq V$:
\begin{align} 
\label{nembounds} f(Y) \leq f(X) - \sum_{j \in X \backslash Y } f(j| X \backslash j) + \sum_{j \in Y \backslash X} f(j| X \cap Y), \\
\label{nembounds2} f(Y) \leq f(X) - \sum_{j \in X \backslash Y } f(j| X \cup Y \backslash j) + \sum_{j \in Y \backslash X} f(j | X) 
\end{align}
Using submodularity, we can loosen these bounds further to provide tight modular 
upper bounds~\cite{aa09,jegelka2009-esc-nips,jegelkacvpr, rkiyersubmodBregman2012, rkiyeruai2012}:
\looseness-1
\begin{align}
\label{modnembounds1} f(Y) \leq f(X) - \sum_{j \in X \backslash Y } f(j | X \setminus \{j\}) + \sum_{j \in Y \backslash X} f(j | \emptyset)  \\
\label{modnembounds2} f(Y) \leq f(X) - \sum_{j \in X \backslash Y } f(j | V \setminus \{j\}) + \sum_{j \in Y \backslash X} f(j | X) \\
\label{modnembounds3} f(Y) \leq f(X) - \sum_{j \in X \backslash Y } f(j | V \setminus \{j\}) + \sum_{j \in Y \backslash X} f(j | \emptyset). 
\end{align}
From the three bounds above, and substituting the expressions of the supergradients, we may
immediately verify that these are supergradients, namely that $\ggrow_X, \gshrink_X, \gbar_X \in \partial^f(X)$. For example, starting with Eqn.~\eqref{modnembounds1}, we have
that for all $Y \subseteq V$:
\begin{align}
f(Y)  
&\leq f(X) - \sum_{j \in X \backslash Y } f(j | X \setminus \{j\}) + \sum_{j \in Y \backslash X} f(j | \emptyset)  \\
&= f(X)
 - \sum_{j \in X } f(j | X \setminus \{j\}) 
 + \sum_{j \in X \cap Y } f(j | X \setminus \{j\}) 
+ \sum_{j \in Y \backslash X} f(j | \emptyset) \\
&= f(X) - \ggrow_X(X) + \ggrow_X(Y) = 
\growsymb{m}^X(Y),
\end{align}
where $\growsymb{m}^X(Y)$ is the modular upper
bound of $f$ associated with the supergradient $\ggrow_X$
that is tight at $Y=X$.
Similar expansions can start with Eqns.~\eqref{modnembounds2}
and~\eqref{modnembounds3}
which define
$\shrinksymb{m}^X(Y)$ 
and 
$\barsymb{m}^X(Y)$ as the modular upper bounds of $f$,
tight at $Y=X$, associated with the supergradients
$\gshrink_X$ and $\gbar_X$ respectively.
\end{proof}

These three supergradients 
(i.e., Eqns.~\eqref{eq:ggrow_def}--\eqref{eq:gbar_def})
can be used to characterize useful and
practical inner bounds of the superdifferential.
First, we define two helper polyhedra:
\begin{align}
\partial^f_{\emptyset}(X) &\triangleq \{x \in \mathbb{R}^n: f(j) \leq x(j), \forall j \notin X\}, \\
\partial^f_V(X) &\triangleq \{x \in \mathbb{R}^n: f(j | V \backslash j) \geq x(j), \forall j \in X\}.
\end{align}
Then we define the following three polyhedra: 
\rishabh{Jeff}{8/12/2015: we should probably define different names for the following
three polyhedra. I.e., rather than using $i$ as a marker for inner (which is
odd since $i$ is normally an index rather than a letter in text mode) , and then using
the numbers 1, 2, and 3, we should just use the same symbols
we used to distinguish grow, shrink, and bar, i.e., $\growsymb{}$, $\shrinksymb{}$,
and $\barsymb{}$. 
Rather than 
using
$\partial^f_{i, 1}(X)$,
$\partial^f_{i, 2}(Y)$,
and $\partial^f_{i, 3}(Y)$,
lets use the macros: 
$\inpolygrow(X)$,
$\inpolyshrink(X)$, and
$\inpolybar(X)$. I made the change below (as well as fixed a bunch of typos
where $X$ and $Y$ were swapped).
}
\begin{align}
\inpolygrow(X)
&\triangleq \partial^f_1(X) \cap \partial^f_{\emptyset}(X) 
= \{x \in \mathbb{R}^n: f(j | X \backslash j) \geq x(j), \forall j \in X \text{ and } f(j) \leq x(j), \forall j \notin X\},
\label{eq:innerbound_grow}
\\ 
%
\inpolyshrink(X)
&\triangleq \partial^f_2(X) \cap \partial^f_{V}(X) 
= \{x \in \mathbb{R}^n: 
f(j | X) \leq x(j), \forall j \notin X \text{ and }  f(j | V \backslash j) \geq x(j), \forall j \in X \},
\label{eq:innerbound_shrink}
\\
%
\inpolybar(X)
&\triangleq \partial^f_V(X) \cap \partial^f_{\emptyset}(X) 
= \{x \in \mathbb{R}^n: f(j | V \backslash j) \geq x(j), \forall j \in X \text{ and } f(j) \leq x(j), \forall j \notin X\}.
\label{eq:innerbound_bar}
\end{align}
Then note that $\inpolygrow(X)$ is a polyhedron with $\ggrow_X$ as an
extreme point. Similarly $\inpolyshrink(X)$ has $\gshrink_X$,
while $\inpolybar(X)$ has $\gbar_X$, as their respective extreme
points. All these are simple polyhedra, each with a single extreme
point. We also define the polyhedron:
\begin{align}
\inpolygrowshrink(X)
 \triangleq \text{conv}(\inpolygrow(X), \inpolyshrink(X))
\end{align}
where $\text{conv}(., .)$ represents the convex combination of two
polyhedra\footnote{Given two polyhedra $\mathcal P_1, \mathcal P_2$,
  $\mathcal P = \text{conv}(\mathcal P_1, \mathcal P_2) = \{\lambda
  x_1 + (1 - \lambda) x_2, \lambda \in [0, 1], x_1 \in \mathcal P_1,
  x_2 \in \mathcal P_2\}$}. Then $\inpolygrowshrink(X)$ is a
polyhedron which has $\ggrow_X$ and $\gshrink_X$ as its extreme
points. The following lemma then characterizes these
polyhedra and how they are inner bounds of the
superdifferential:
\begin{lemma}[Superdifferential Inner Bound Relationships]
\label{thm:super_innerbound_relations}
Given a submodular function $f$, 
\begin{align}
\inpolybar(X) \subseteq \inpolygrow(X) \subseteq \inpolygrowshrink(X) \subseteq \partial^f(X), \\
\inpolybar(X) \subseteq \inpolyshrink(X) \subseteq \inpolygrowshrink(X) \subseteq \partial^f(X).
\end{align}
\end{lemma}
\begin{proof}
  The proof of this lemma follows directly from the definitions of the
  supergradients, the corresponding polyhedra, and of submodularity.
\end{proof}
 
\begin{figure}[tbh]
\centering{
\includegraphics[width = 0.6\textwidth,]{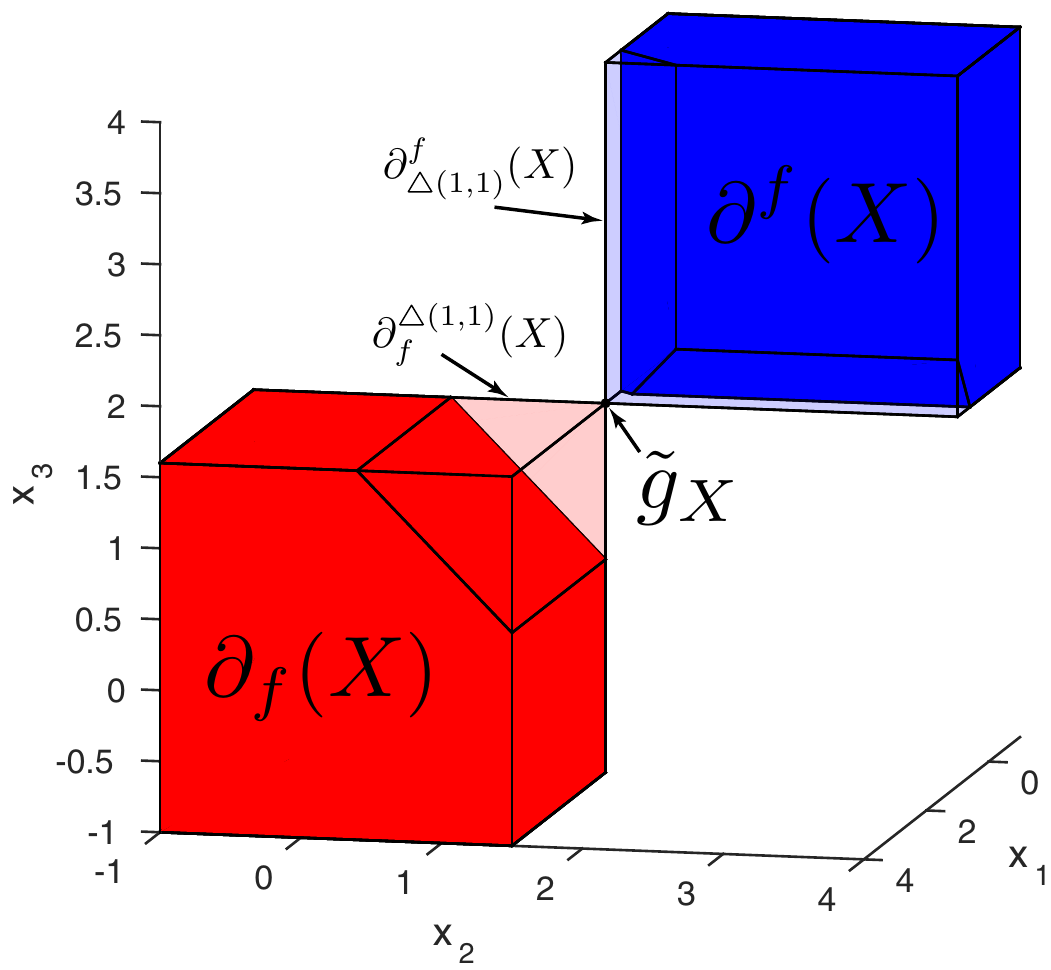}
}
\caption{An illustration to compare the relative positions of the sub-
  and superdifferentials on the submodular function $f:
  2^{\{1, 2, 3\}} \rightarrow \mathbb{R}$ defined 
  as shown in Table~\ref{tab:submod_instance_1}.
The 
  subdifferentials appear in red, while the superdifferential is shown in
  blue and are defined at $X = \{1\}$.  Also shown
  is the point $\tilde g_X$ defined in Eqn.~\eqref{eq:gtilde} and
  the corresponding outer bounds 
  ($\partial_f^{\symmdiff(1, 1)}(X)$ and
$\partial^f_{\symmdiff(1, 1)}(X)$ defined in Eqns~\eqref{localsub} and
\eqref{localsup}) of the two semidifferentials.
  \rishabh{Jeff}{I think this
    figure is probably not correct as the supergradients don't have
    this shape if the axes are standard. Should fix
    this.}\jeff{Rishabh}{Added a real example here.}
  \rishabh{Jeff}{Make vector version}\jeff{Rishabh}{Added a vector
    version above.}
   \rishabh{Jeff}{Thanks. I've annotated the figure a bit more in a new figure pdf and added more
     comments as well as fixed many typos.}
}
\label{subsupdiffrel}
\end{figure}

\subsubsection{Connections between the subdifferential and superdifferential at $X$}
\label{sec:conn-betw-subd}

There are some interesting connections between $\partial_f(X)$ and
$\partial^f(X)$. Firstly, it is clear from the definitions that
$\partial_f(X) \subseteq \partial_f^{\symmdiff(1, 1)}(X)$ and
$\partial^f(X) \subseteq \partial^f_{\symmdiff(1, 1)}(X)$. Notice also
that both $\partial_f^{\symmdiff(1, 1)}(X)$ and
$\partial^f_{\symmdiff(1, 1)}(X)$ (from Eqns.~\eqref{localsub}
and~\eqref{localsup} respectively) are simple polyhedra containing a
single extreme point $\tilde{g}_X \in \mathbb R^n$ defined as follows:
\begin{align}
\tilde{g}_X(j) \triangleq 
\begin{cases}
f(j |X \setminus j) & \text{ if }  j \in X\\
f(j | X) & \text { if } j \notin X\\
\end{cases}
\label{eq:gtilde}
\end{align}
The point $\tilde{g}_X$ is, in general, neither a subgradient nor a
supergradient at $X$. Each of the semidifferentials, $\partial_f(X)$
and $\partial^f(X)$, however, are contained within a (distinct)
polyhedra defined via $\tilde{g}_X$. In particular, 
$\tilde{g}_X \in \partial_f^{\symmdiff(1, 1)}(X)$
and 
$\tilde{g}_X \in \partial^f_{\symmdiff(1, 1)}(X)$,
and $\tilde{g}_X$ is
an extreme point both of $\partial_f^{\symmdiff(1, 1)}(X)$ and
$\partial^f_{\symmdiff(1, 1)}(X)$.
An illustration of this is in
Figure~\ref{subsupdiffrel}. The subdifferential $\partial_f(X)$ is the
red polyhedron, while the superdifferential $\partial^f(X)$ is the
blue polyhedron. Moreover, the light red and the light blue polyhedra
are $\partial_f^{\symmdiff(1, 1)}(X)$ and $\partial^f_{\symmdiff(1,
  1)}(X)$, respectively, defined at $X = \{1\}$.

Since $\partial^f_{\symmdiff(1, 1)}(X)$ is not a superdifferential,
but rather an outer bound on one, the modular function
\begin{align}
\tilde m^X(Y) \triangleq f(X) - \tilde g_X(X) + \tilde g_X(Y)
\end{align}
is not everywhere a modular upper bound on the submodular function
$f$, although it is tight at $Y=X$. If we consider, however, a subset
$[\emptyset , X ] \cup [X , V ] = \{ Y \in 2^V : Y \subseteq X \text{
  or } Y \supseteq X \}$ of sets, then a supergradient property is
retained.
\begin{lemma}
\label{thm:supergrad_subset_supersets}
Given $X \subseteq V$, and any $Y \in [\emptyset , X ] \cup [X , V ]$, then
\begin{align}
  f(Y) \leq f(X) - \tilde g_X(X) + \tilde g_X(Y)
\end{align}
\begin{proof}
Suppose $Y \subseteq X$, then
\begin{align}
  f(X) - f(Y)
=   f( X \setminus Y | Y)  \geq \sum_{j \in X \setminus Y} f(j | X \setminus j)
\end{align}
If, on the other hand, $Y \supseteq X$, then
\begin{align}
  f(Y) - f(X) 
= f(Y \setminus X | X) \leq \sum_{j \in Y \setminus X} f(j | X)
\end{align}
Combining the two together yields:
\begin{align}
f(Y) &\leq f(X) + \sum_{j \in Y \setminus X} f(j|X) - \sum_{j \in X \setminus Y} f(j | X \setminus j) \\
    &= f(X) + \tilde g_X(Y \setminus X) - \tilde g_X(X \setminus Y)
 = f(X) + \tilde g_X(Y) - \tilde g_X(X)
\end{align}
\end{proof}
\end{lemma}


\subsubsection{Examples of inner and outer bounds for specific
  superdifferentials}
\label{sec:examples-inner-outer}

\begin{wraptable}[14]{R}{0.2\textwidth}
\vspace{-0.9\baselineskip}
\begin{tabular}{ |c | c| } \hline
 $X$ & $f(X)$      \\ \hline\hline
 $\emptyset$ & $0$ \\ \hline
 $\{1\}$ & $1$ \\ \hline
 $\{2\}$ & $2$ \\ \hline
 $\{3\}$ & $2$ \\ \hline
 $\{1, 2\}$ & $2.5$ \\ \hline
 $\{2, 3\}$ & $3$ \\ \hline
 $\{1, 3\}$ & $2.8$ \\ \hline
 $\{1, 2, 3\}$ & $3$ \\ \hline
\end{tabular}
\vspace{-0.4\baselineskip}
\caption{An illustrative submodular function defined on $V = \{1,2,3\}$.}
\label{tab:submod_instance_1}
\end{wraptable}

We next investigate the inner and outer bounds of specific
instances. First, consider the superdifferential 
$\partial^f(\emptyset)$
at the empty set $X=\emptyset$. In this case, notice that all three
supergradients are the same vector, i.e $\ggrow_{\emptyset} =
\gshrink_{\emptyset} = \gbar_{\emptyset}$, with the individual elements
being $\ggrow_{\emptyset}(j) = \gshrink_{\emptyset}(j) =
\gbar_{\emptyset}(j) = f(j), j \in V$. Therefore, in this case, the inner 
bounds
(Eqns.~\eqref{eq:innerbound_grow}--\eqref{eq:innerbound_bar})
are
exactly the superdifferential itself, and 
$\inpolygrow(\emptyset) = \inpolyshrink(\emptyset) =
\inpolybar(\emptyset) = \inpolygrowshrink(\emptyset)
= \partial^f(\emptyset)$.  
Also, the largest outer bound $\partial^f_{\symmdiff(1, 1)}(\emptyset)$
from Eqn.~\eqref{localsup}
has the relationship $\partial^f_{\symmdiff(1, 1)}(\emptyset)
= \partial^f(\emptyset)$ since  
$\tilde{g}_{\emptyset}$
is also identical to these supergradients.
Therefore, in this case, all of the inner
and outer polyhedral bounds are identical to the
superdifferential. This phenomenon also occurs for the
superdifferential at the ground set $\partial^f(V)$. For other sets,
however, this does not hold and the relationship between the inner and
outer bounds and the superdifferential can be strict.

\begin{figure}[tbh]
\centering{
\includegraphics[width = 0.9\textwidth]{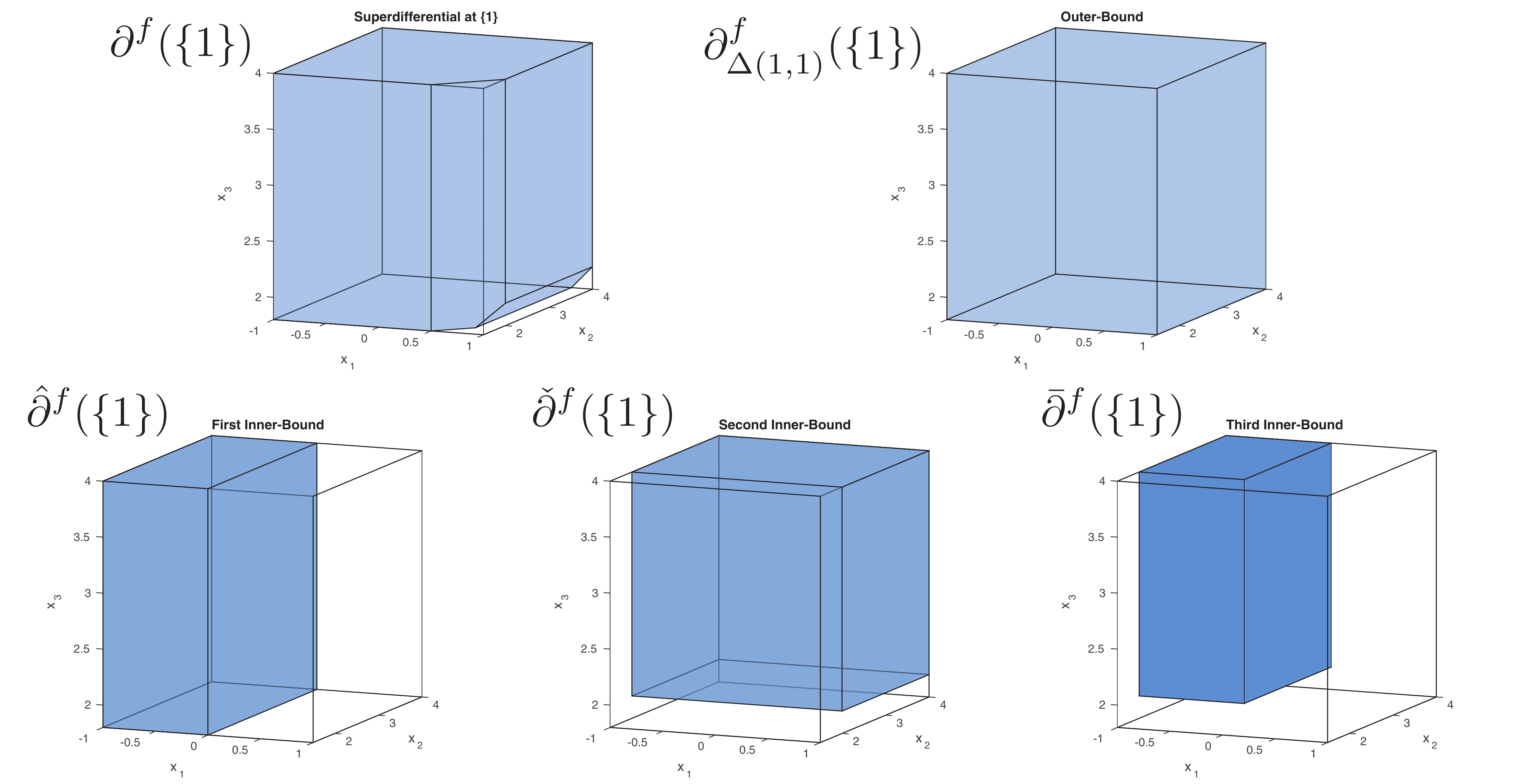}
}
\caption{A visualization of the inner and outer bounds of the superdifferential. The submodular function is given in
Table~\ref{tab:submod_instance_1}.
The shown superdifferential $\partial^f(X)$ is at $X = \{1\}$. The first
figure (top left) shows the submodular supergradient $\partial^f(X)$ itself,
while the second one (top right) is the outer bound $\partial^f_{\Delta(1,
  1)}(X)$. The bottom three figures show the inner bounds
$\inpolygrow(\{1\})$,
$\inpolyshrink(\{1\})$,
and 
$\inpolybar(\{1\})$.
The polyhedral inner bound
$\inpolygrowshrink(\{1\}) = 
\text{conv}(\inpolygrow(\{1\}), \inpolyshrink(\{1\}))$
is not shown.
\rishabh{Jeff}{Fix this figure to show a real instance in the 2D case,
  and give the corresponding submodular
  function.}}\jeff{Rishabh}{Added a real example!}\rishabh{Jeff}{can
you please make this vector. Also, I'm not getting a good 3D sense
from this figure, can you try rotating it, or something with
the colors?}\jeff{Rishabh}{I added two figures above, which are both
vectors, and give different views of this.}
\jeff{Rishabh}{I've changed the figure to use the new notation.}
\label{3douterinnersupervis}
\end{figure}

In the next example, we analyze the inner and outer bounds of the
superdifferentials for some specific submodular functions in order to
get further intuition about them.

%

\begin{example}
\label{ex3}
In this example, we show how in 2-D some of the inner and
outer bounds are exact, and other of the inner bounds (resp.\
outer bounds) are strictly smaller (resp.\ larger) than their
corresponding exact superdifferentials.
Assume the ground set is $V = \{1,2\}$. From Lemma~\ref{thm:super_at_empty_and_full}, we know that the superdifferentials $\partial^f(\emptyset)$ and $\partial^f(\{1, 2\})$ are simple polyhedra and, as mentioned 
at the beginning of Section~\ref{sec:examples-inner-outer},
the inner and outer bounds are identical to the superdifferential itself. 

Consider, however, $\partial^f(\{1\})$. Recall from Example~\ref{ex1}
that the extreme points here are $\{f(\{1\}), f(\{2\})\}$ and
$\{f(\{1\} | \{2\}), f(\{2\} | \{1\})\}$ respectively. Notice that
$\ggrow = (f(\{1\}), f(\{2\}))$ and $\gshrink = (f(\{1\} | \{2\}),
f(\{2\} | \{1\}))$, and hence both of these supergradients are extreme
points of the superdifferential in two dimensions. Also note that
$\gbar = (f(\{1\} | \{2\}), f(\{2\}))$. Therefore, $\gbar$ lies in
the interior of the superdifferential for a strictly submodular
function\footnote{A strict submodular function is a submodular
  function where none of the defining inequalities act as equalities.}
since (considering inequalities in Eqns.~\eqref{eq1d2}--\eqref{eq3d2}
governing $\partial^f(\{1\})$ from Example~\ref{ex1}), we have $x_1 =
f(\{1\} | \{2\}) < f(\{1\})$, $x_2 = f(\{2\}) > f(\{2\} | \{1\})$, and
$x_1 - x_2 = f(\{1\} | \{2 \}) - f(\{2\}) = f(\{1,2\}) - 2f(\{2\}) <
f(\{1\}) - f(\{2\})$ (which follows since $f(\{1\}, \{2\}) < f(\{1\})
+ f(\{2\})$).  In this case, therefore,
Lemma~\ref{thm:super_innerbound_relations} becomes
%
$\inpolybar(\{1\}) \subset \inpolyshrink(\{1\}) \subset \inpolygrowshrink(\{1\}) = \partial^f(\{1\})$
and 
$\inpolybar(\{1\}) \subset \inpolygrow(\{1\}) \subset \inpolygrowshrink(\{1\}) = \partial^f(\{1\})$.

Similarly, observe from Eqn.~\eqref{eq:gtilde} that $\tilde{g}_{\{1\}}
= (f(\{1\}), f(\{2\} | \{1\}) \in 
\partial^f_{\symmdiff(1, 1)}(\{1\})$ 
does not belong to $\partial^f(\{1\})$
when $f$ is strictly submodular since it violates Eqn.~\eqref{eq3d2},
i.e., $f(\{1\}) - f(\{2\} | \{1\}) \leq f(\{1\}) - f(\{2\})$ (this
does not hold since it would require $f(\{2\} | \{1\}) \geq f(\{2\})$
which violates strict submodularity). Hence $\partial^f_{\symmdiff(1,
  1)}(\{1\}) \supset \partial^f(\{1\})$. The same phenomena is true
for $\partial^f(\{2\})$.
\end{example}

We can also consider the superdifferential in the three dimensional
setting when $V = \{1, 2, 3\}$. In this case we must consider an
specific submodular function instance, and this
is done in Table~\ref{tab:submod_instance_1}.
Consider $\partial^f(\{1\})$. An illustration of
this is in Figure~\ref{3douterinnersupervis}. The 
various polyhedra are shown shaded in blue.
Note that the rectangular polyhedron (upper right case)
is the outer bound $\partial^f_{\symmdiff(1,
  1)}(X)$. In this case, 
it holds that 
\begin{align}
\inpolybar(\{1\}) 
\subset \inpolygrow(\{1\}) 
\subset \inpolygrowshrink(\{1\}) 
\subset \partial^f(\{1\}) \subset \partial^f_{\symmdiff(1,1)}(\{1\})
\intertext{and}
\inpolybar(\{1\}) 
\subset \inpolyshrink(\{1\}) 
\subset \inpolygrowshrink(\{1\}) 
\subset \partial^f(\{1\}) \subset \partial^f_{\symmdiff(1,1)}(\{1\})
\end{align}
That is, the subset relationships are strict in this case.

\subsubsection{Superdifferentials of subclasses of submodular functions}
\label{sec:superd-subcl-subm}

While it is hard to characterize superdifferentials of general
submodular functions, certain subclasses have easy
characterizations. An important such subclass of the class of
$M^{\natural}$-concave\footnote{In this paper, we consider only those
  $M^{\natural}$-concave functions defined on $2^V = \{ 0,1\}^V$ while
  $M^{\natural}$-concave functions are typically
  defined~\cite{murota2003discrete} on $\mathbb Z^V$.}
functions~\cite{murota2003discrete} defined on $2^V$. 
These include a number of special
cases like matroid rank functions, concave over cardinality functions
etc. All $M^{\natural}$-concave functions defined on $2^V$ are submodular on $2^V$
but not vice verse. 
In some sense, $M^{\natural}$-concave functions very closely resemble concave
functions. In particular, one can maximize these functions exactly in
polynomial time~\cite{murota2003discrete}. These functions also admit
simple characterizations of their superdifferential. In particular, the
superdifferential of this class of functions can be represented using
only $O(n^2)$ inequalities. The following theorem provides a compact
representation of the superdifferential of these functions.
\begin{lemma}
Given a submodular function $f$ which is also $M^{\natural}$-concave on $\{0,1\}^V$, its superdifferential satisfies:
\begin{align}
\partial^f(X) = \partial^f_{\symmdiff(2, 2)}(X).
\end{align}
In particular, it can be characterized via $O(n^2)$ inequalities.
\end{lemma}
\begin{proof}
A set function $\mu$ is said to be $M^{\natural}$ concave~\cite{murota2008submodular}, if for any $X, Y \subseteq V$ and any $i \in X \backslash Y$, we have that either the following inequality is true:
\begin{align}
\mu(X) + \mu(Y) \leq \mu(X \backslash \{ i \} ) + \mu(Y \cup \{ i \}),
\end{align}
or if not, then there is some $j \in Y \backslash X$ where:
\begin{align}
\mu(X) + \mu(Y) \leq \mu( (X \backslash \{i \}) \cup \{j\}) + \mu( (Y \cup \{i\}) \backslash \{j\}).
\end{align}
This is called the exchange property for said functions.  We then
invoke Theorem 6.61 in \cite{murota2003discrete} where the authors
show that for a $M^{\natural}$ {\em convex} function (which is supermodular
on $2^V$ and is defined using the opposite inequality to the
above), its subdifferential (which in fact corresponds
to a superdifferential of a submodular function)
can be expressed by just considering sets
$Y$ satisfying $|X \backslash Y| \leq 1, |Y \backslash X| \leq 1$ (i.e.,
of Hamming distance less than two). In particular, we have that,
\begin{align}
\partial^{\mu}(X) = \{x \in \mathbb{R}^n: &x(j) \leq \mu(j | X \backslash j), \forall j \in X\\
			 							& x_j \geq \mu(j | X), \forall j \notin X \\
										& x_i - x_j \leq \mu(X) - \mu(X \cup j \backslash i), \forall i \in X, j \notin X\}
\end{align}
Hence the superdifferential of a $M^{\natural}$ concave function (which is submodular) can be expressed with the same number of inequalities and the corresponding polyhedron is $\partial^f_{\symmdiff(2, 2)}(X)$. 
\rishabh{Jeff}{3/31/2015: Please give more details for this proof. I.e. give a complete definition of $M^{\natural}$ convex functions in the binary case. I think we should do this a bit more cavalierly just in case something
is not right, as there is a bit of subtlety when defining $M^{\natural}$ convex functions. 
See in particular his \url{http://www.kurims.kyoto-u.ac.jp/~kenkyubu/bessatsu/open/B23/pdf/B23-10.pdf}
paper section 5.1, page 203 for an easy definition to use.}\jeff{Rishabh}{Yes, I have added a lot more details above.}
\end{proof}
\rishabh{Jeff}{
This is pretty powerful. What would be useful to show is
if $\partial^f(X) = \partial^f_{\symmdiff(2, 2)}(X)$ is a necessary and sufficient condition
for $M^{\natural}$ concave functions (and we'd need a full proof for this). Moreover, for other fixed $\ell > 2$,
the condition $\partial^f(X) = \partial^f_{\symmdiff(\ell,\ell)}(X)$ might define a chain of subclasses of submodular
functions that can be maximized in polynomial time (for fixed $\ell$).
}\jeff{Rishabh}{Yes, I think this is an interesting problem. I have added a description of this as an open problem towards the end.}



\subsection{Generalized Submodular Upper Polyhedron}
\label{sec:gener-subm-upper}


In this section, we generalize the submodular upper polyhedron from
Section~\ref{uppersubpolysec} in a manner analogous to how the
generalized submodular lower polyhedron of Section~\ref{gensubmodpoly}
generalized the submodular (lower) polyhedron of
Section~\ref{submodpoly}. Unlike for the submodular lower polyhedron case,
however, for the generalized submodular upper polyhedron some real
utility will ensue.

We define the \emph{generalized submodular upper polyhedron} as the set
of affine upper bounds of $f$ as follows:
\begin{align}
\mathcal P^f_{\text{gen}} \triangleq \{(x, c), x \in \mathbb{R}^n, c \in \mathbb{R}: x(X) + c \geq f(X), \forall X \subseteq V\}
\label{eq:gen_sub_up_poly}
\end{align}
Again it is easy to see that $\mathcal P^f_{\text{gen}} \cap \{(x, c): c = 0\} = \{(x, c): x \in \mathcal P^f, c = 0\}$. In other words, the slice $c = 0$ of the \emph{generalized submodular upper polyhedron} is the submodular upper polyhedron of $f$. Also note that the inequality at $X = \emptyset$ implies that $c \geq 0$.
This polyhedron shall prove to be useful while defining concave extensions of $f$. The generalized submodular upper polyhedron also has interesting connections with the superdifferentials. In particular, we have the following:
\begin{lemma}
\label{suppdiffandantipolyrel1}
Given a submodular function $f$, $(x, c) \in \mathcal P^f_{\text{gen}}$ lies on a face of the polyhedron if and only if there exists a set $X$ such that $x \in \partial^f(X)$ and $c = f(X) - x(X)$.
 \end{lemma}
\begin{proof}
  The proof of this lemma is analogous to the one for the generalized
  submodular lower polyhedron in Lemma~\ref{lemma:face_gen_lower_poly}.
 In particular, observe that $(x, c)$
  lies on a face of $\mathcal P^f_{\text{gen}}$ if and only if there
  exists a set $X$ such that $x(X) + c = f(X)$ and for all $Y
  \subseteq V, x(Y) + c \geq f(Y)$. It then directly implies that $x
  \in \partial^f(X)$ and $c = f(X) - x(X)$.
\end{proof}  
This then implies the following corollary:
 \begin{corollary}
Given a submodular function $f$, a point $(x, c)$ is an extreme point of $\mathcal P^f_{\text{gen}}$, if and only if $x$ is an extreme point of $\partial^f(X)$ for some set $X$
and $c = f(X) - x(X)$. 
\label{suppdiffandantipolyrel1corr}
\end{corollary}
\begin{proof}
Assume that $(x, c)$ is an extreme point of $\mathcal P^f_{\text{gen}}$. Then, there must be $n+1$ sets $X_0, X_1, \cdots, X_n$ such that $x(X_i) + c = f(X_i), \forall i = 0, 1, 2, \cdots, n$, and $x(X) + c \geq f(X), \forall X \subseteq V$. Set $c = f(X_0) - x(X_0)$. This implies that $x(X_i) + f(X_0) - x(X_0) = f(X_i), \forall i = 1, 2, \cdots, n$, and $x(X) + f(X_0) - x(X_0) \geq f(X_0), \forall X \subseteq V$. This implies that $x$ is an extreme point of $\partial^f(X_0)$. To prove the other direction, we start with a set $X_0$, such that $x$ is an extreme point of $\partial^f(X_0)$. Set $c = f(X_0) - x(X_0)$. Then, following the fact that $x$ is an extreme point of $\partial^f(X_0)$, we know that there exist $n$ sets $X_1, \cdots, X_n$ such that $x(X_i) + f(X_0) - x(X_0) = f(X_i), \forall i = 1, 2, \cdots, n$, and $x(X) + f(X_0) - x(X_0) \geq f(X_0), \forall X \subseteq V$. Substituting for $c$, we observe that $x(X_i) + c = f(X_i), \forall i = 0, 1, 2, \cdots, n$, and $x(X) + c \geq f(X), \forall X \subseteq V$. This proves that $(x, c)$ is an extreme point of $\mathcal P^f_{\text{gen}}$.
\end{proof}
\jeff{Rishabh}{AUGTODO: 8/19/2015: We need to have a proof of this here.}\rishabh{Jeff}{Hi Jeff, I went through this proof, and for some reason I am having trouble proving it now. I have currently removed this corollary. Fortunately, I think the result below, however, does not depend on this.}
\jeff{Rishabh}{9/3/2015: We talked about it today, and came up with a proof, so can you please add it here?}

This implies an interesting characterization of a linear program over the generalized submodular upper polyhedron. 
 \begin{lemma}\label{suppdiffandantipolyrel2}
 For submodular function $f$, and a $y \in \mathbb{R}^n$,
 \begin{align}\label{suppdiffandantipolyreleq}
 \min_{(x, c) \in \mathcal P^f_{\text{gen}}} [ \langle x, y \rangle + c ] 
  = \min\{\min_{x \in \partial^f(X)} [ \langle x, y \rangle + f(X) - x(X)] \,\, | \,\ X \subseteq V\}.
 \end{align}
 \end{lemma}
 \begin{proof}
We first show that $\min_{(x, c) \in \mathcal P^f_{\text{gen}}} [ \langle x, y \rangle + c ]\leq \min\{\min_{x \in \partial^f(X)} \langle x, y \rangle + f(X) - x(X)\,\, | \,\ X \subseteq V\}$. Observe that for any set $X$, and point $x \in \partial^f(X)$, $(x, f(X) - x(X)) \in \partial^f(X)$. Hence the second expression can be obtained by taking only a subset of the polyhedron $\mathcal P^f_{\text{gen}}$, and hence is a upper bound. 
Next, we show that $\min_{(x, c) \in \mathcal P^f_{\text{gen}}} [ \langle x, y \rangle + c ] \geq \min\{\min_{x \in \partial^f(X)} \langle x, y \rangle + f(X) - x(X)\,\, | \,\ X \subseteq V\}$ by 
invoking Lemma~\ref{suppdiffandantipolyrel1}. The minimum on the l.h.s.\ must occur at an extreme point of $\mathcal P^f_{\text{gen}}$, which implies that $x \in \partial^f(X)$ for some set $X$, and $c = f(X) - x(X)$. Hence this implies that $\min_{(x, c) \in \mathcal P^f_{\text{gen}}} \langle x, y \rangle + c \geq \min\{\min_{x \in \partial^f(X)} \langle x, y \rangle + f(X) - x(X)\,\, | \,\ X \subseteq V\}$, since the l.h.s.\ equals a particular instance of the RHS. This completes the proof.
 \end{proof}
 
Unfortunately, however, the generalized submodular upper polyhedron is no longer easy to characterize. This is related to the fact that the superdifferentials of a submodular function are not easy to characterize. 
 
 \begin{lemma}\label{NPgenantisubpoly}
   The generalized submodular upper polyhedron membership problem for
   a submodular function $f$ (i.e., given an $x \in \mathbb{R}^n$ and $c \in
   \mathbb{R}$, solve the query ``Is $(x, c) \in \mathcal
   P^f_{\text{gen}}$?'') is NP hard for $c > 0$. Furthermore, for any $y
   \in \mathbb{R}^n$, solving a linear program over this polyhedron,
   i.e $\min_{(x, c) \in \mathcal P^f_{\text{gen}}} \langle x, y
   \rangle + c$ is also NP hard.
\end{lemma}
\begin{proof}
The first part of the result follows from the fact that asking whether $(x, c) \in \mathcal P^f_{\text{gen}}$ is equivalent to asking whether $\max_{X \subseteq V} [f(X) - x(X) - c] \leq 0$, which can be rewritten as $\max_{X \subseteq V} [f(X) - x(X)] \leq c$. This is the decision version of submodular maximization, which is NP hard. The second part follows directly from the first since the membership problem on a polyhedron is equivalent to a linear program over this polyhedron~\cite{grotschel1984geometric, schrijver2003combinatorial}.
 \end{proof} 
We can also prove the second part (that solving a linear program over the generalized submodular polyhedron is NP hard) since it is equivalent to computing the concave extension of a submodular function (we show this in Lemma~\ref{lemma:lpconcaveextension}). Computing, and in fact even evaluating at a point, this concave extension, however, is NP hard~\cite{dughmi2009submodular, vondrak2007submodularity}.
 
Recall that in the case of the generalized submodular lower
polyhedron, the extreme points of this polyhedron were identical to
the extreme points of the submodular lower polyhedron (i.e., all
extreme points of the generalized submodular lower polyhedra occurred
when $c = 0$) --- this that the linear program over the two
polyhedra was the same. This is not the case in the generalized
submodular upper polyhedra. To see this, we consider a simple example
with $V = \{1, 2\}$.

\begin{figure}[tbh]
\begin{center}
\includegraphics[width = 0.7\textwidth]{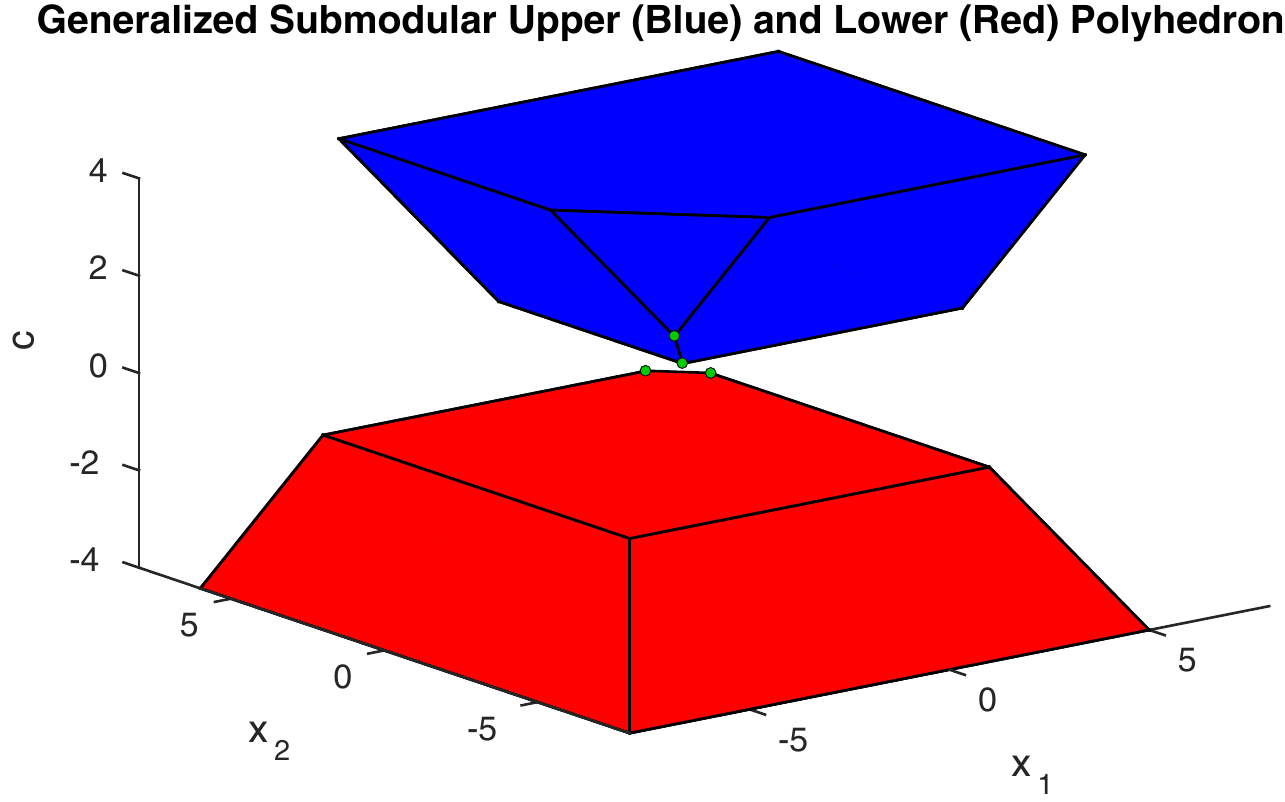}
\end{center}
\caption{The top figure shows the 
generalized submodular upper polyhedron 
(Eqn.~\eqref{eq:gen_sub_up_poly})
in blue, 
while the bottom shows 
the generalized lower polyhedron 
(Eqn.~\eqref{eq:gen_sub_low_poly})
in red for a submodular function $f: 2^{\{1, 2\}} \rightarrow \mathbf{R}$, with $f(\emptyset) = 0, f(\{1\}) = 1, f(\{2\}) = 2, f(\{1, 2\}) = 2.5$. The polyhedra live
in three dimensions for two-dimensional submodular functions.
Notice that all the extreme points (green) of the generalized submodular lower polyhedron are on the plane $c = 0$ -- the two extreme points are $(1, 1.5, 0)$ and $(0.5, 2, 0)$. In the generalized submodular upper polyhedron, however, one of the extreme points is on $c = 0$ (this extreme point is $(1, 2, 0)$), while the other extreme point is $(0.5, 1.5, 0.5)$ (here $c = 0.5 > 0$).
\rishabh{Jeff}{8/19/2015: also we should mark on the figure precisely where
the extreme points are.}\rishabh{Jeff}{this is now done.}
}
\label{genupperlowerfigs}
\end{figure}

\begin{example}\label{ex4}
First consider the generalized submodular lower polyhedra when $V = \{1, 2\}$. 
\begin{align}
\mathcal P^f_{\text{gen}} = \{(x, c) \in \mathbb{R}^3: &c \leq 0, \label{eqex1} \\ 
			&x_1 + c \leq f(\{1\}), \label{eqex2}\\
			 							& x_2  + c\leq f(\{2\}), \label{eqex3}\\
										& x_1 + x_2 + c \leq f(\{1, 2\})\} \label{eqex4}
\end{align}
It is immediate that the only extreme points are $(f(\{1\}), f(\{2\}
| \{1\}), 0)$ and $(f(\{1\} | \{2\}), f(\{2\}), 0)$, which are
obtained by setting Eqns~\eqref{eqex1}, \eqref{eqex2}, \eqref{eqex4}
and Eqns~\eqref{eqex1}, \eqref{eqex3}, \eqref{eqex4} as
equalities. The extreme points in this case are a direct product
between the extreme points of $\mathcal P_f$ and $c = 0$. Hence all
extreme points lie on the face $c = 0$.
\end{example}
This is not the case for the generalized submodular upper polyhedron. Consider again the example with $V = \{1, 2\}$.
\begin{example}\label{ex5}
The generalized upper submodular polyhedron in this case is, 
\begin{align}
\mathcal P^f_{\text{gen}} = \{(x, c) \in \mathbb{R}^3: &c \geq 0, \label{eqex11} \\ 
			&x_1 + c \geq f(\{1\}), \label{eqex12}\\
			 							& x_2  + c\geq f(\{2\}), \label{eqex13}\\
										& x_1 + x_2 + c \geq f(\{1, 2\})\} \label{eqex14}
\end{align}
It is again immediate that the only extreme points are $ ( f(\{1\}), f(\{2\}), 0)$ and \\ $(f(\{1\} | \{2\}), f(\{2\} | \{1\}), f(\{1\}) + f(\{2\}) - f(\{1, 2\}))$, which are obtained by setting Eqns~\eqref{eqex11}, \eqref{eqex12}, \eqref{eqex13} and Eqns~\eqref{eqex12}, \eqref{eqex13}, \eqref{eqex14} as equalities (setting the other combination of inequalities as equalities does not give extreme points). Hence while one of the extreme points here is $\{f(\{1\}), f(\{2\}), 0\}$, which is the direct product between $\mathcal P^f$ and $c = 0$, the other extreme point occurs at $c$, and when $f$ is strictly submodular, $c > 0$.
\end{example}
An illustration of the generalized submodular upper and lower polyhedra is shown in Figure~\ref{genupperlowerfigs}.


\subsubsection{Inner and outer bounds on the generalized submodular upper polyhedron}
\label{sec:inner-outer-bounds}

In a manner similar to the superdifferential, we can provide inner and
outer bounds of the generalized submodular upper polyhedron. In
particular, let $g_X \in \partial^f(X)$ be a supergradient at $X$ that
is feasible to obtain (such as the ones in
Eqns.~\eqref{eq:ggrow_def}--\eqref{eq:gbar_def}).  Then, $m^X(Y) =
f(X) + g_X(Y) - g_X(X)$ is a modular upper bound of $f(Y), \forall Y
\subseteq V$.  Given any set $\mathcal G = \{ g_X \in \partial^f(X) |
X \subseteq V \}$ of such supergradients, we may define a polytope as
follows:
\begin{align}
\mathcal P^f_{\mathcal G, \text{gen}} 
\triangleq \text{conv-hull}\{(g_X, f(X) - g_X(X)), \forall X \subseteq V,
g_X \in \mathcal G
\}.
\end{align}
Since for any $X \subseteq V$ and $g_X \in \partial^f(X)$, we have
that $(g_X, f(X) - g_X(X)) \in P^f_{\text{gen}}$, it follows from the
convexity of $P^f_{\text{gen}}$ that $\mathcal P^f_{\mathcal G,
  \text{gen}} \subseteq \mathcal P^f_{\text{gen}}$.  Moreover, larger
inner bounds of $P^f_{\text{gen}}$ can be obtained by taking the
convex hull of multiple such polytopes of the form $\mathcal
P^f_{\mathcal G, \text{gen}}$ for various $\mathcal G$. We shall in
particular be interested by the polytopes $\growsymb{\mathcal G} = \{
\ggrow_X | X \subseteq V \}$, $\shrinksymb{\mathcal G} = \{ \gshrink_X
| X \subseteq V \}$, and $\barsymb{\mathcal G} = \{ \gbar_X | X
\subseteq V \}$, formed using
Eqns.~\eqref{eq:ggrow_def}--\eqref{eq:gbar_def}), which we will refer
to with $\mathcal P^f_{\growsymb{\mathcal G}, \text{gen}}$, $\mathcal
P^f_{\shrinksymb{\mathcal G}, \text{gen}}$, and $\mathcal
P^f_{\barsymb{\mathcal G},\text{gen}}$.  These bounds, as we shall
see, have interesting connections to concave extensions (which we
shall describe Section~\ref{submodccvext}) and ultimately to
submodular maximization. 

In a fashion analogous to how, in
Section~\ref{sec:outer-bounds-superd}, we defined outer bounds on the
submodular differential, we can similarly define outer bounds of the
generalized submodular upper polyhedron by considering only a subset
of inequalities that define $\mathcal P_{\text{gen}}^f$. We do not
pursue this here and leave it to future work (see
Section~\ref{sec:concl-open-probl}).

\rishabh{Jeff}{8/19/2015:I've changed quite a bit of notation above and
added clarification. In general, don't forget to put the ``text'' macro around the
string ``gen'' in math mode since it is a string.}

\section{Concave extensions of a submodular function}
\label{submodccvext}

Following the characterizations of the convex extensions of a
submodular function, we can define the concave extensions also from
two viewpoints, one in the \emph{distributional} setting and another
in the \emph{polyhedral} setting. These results follow in the lines of
the results shown in
Section~\ref{submodcvxext}
for the convex extensions.

\subsection{Polyhedral characterization of the concave extension}
\label{sec:polyh-char-conc}

Similar to the convex extension, the concave extension of any set
function (not necessarily submodular) can be seen as the pointwise
supremum of concave functions that lower bound the set
function~\cite{dughmi2009submodular}. Precisely, let
\begin{align}
\Psi_f \triangleq \{\psi: \psi \text{ is concave in }[0, 1]^V \text{ and } \psi(1_X) \geq f(X), \forall X \subseteq V\}.
\end{align}
 Then define the concave extension
$\cex f : [0, 1]^{|V|} \to \mathbb R$ as follows:
\begin{align} \label{concaveextccv}
\cex f(w) \triangleq \min_{\psi \in \Psi_f} \psi(w).
\end{align} 
Following arguments similar to the convex extension, Eqn.~\eqref{concaveextccv} can be expressed as a linear program over the generalized submodular upper polyhedron.
\begin{lemma}
\label{lemma:lpconcaveextension}
The concave extension in Eqn.~\eqref{concaveextccv} for any set function $f$ can be expressed as:
\begin{align}
\label{concaveextaff}
\cex f(w) = \min_{(y, c) \in \mathcal P^f_{\text{gen}}} \bigl[ \langle y, w \rangle + c \bigr]
\end{align}
\end{lemma}
\begin{proof}
The proof of this lemma follows the proof of Lemma~\ref{convexextafflemma}. For a given $w$, let $\hat{\psi}$ be an $\argmin$ in Eqn.~\eqref{concaveextccv}. Then since $\hat{\psi}$ is a concave function in $[0, 1]^V$, there exists a supergradient $x \in \mathbb{R}^n$ at $w$ and value $d$, such that $\langle x, y \rangle + d \geq \hat{\psi}(y), \forall y$ and $\langle x, w \rangle + d = \hat{\psi}(w)$. In other words, $\langle x, y \rangle + d$ is a linear upper bound of $\hat{\psi}(y)$, tight at $w$. Hence $\cex f(w) = \langle x, w \rangle + d$. Finally notice that $(x, d) \in \mathcal P^f_{\text{gen}}$ since $x(X) + d \geq \hat{\psi}(1_X) \geq f(X), \forall X \subseteq V$.
\end{proof}
Unlike the case 
shown in Lemma~\ref{linprogsubpoly}
for the convex extension, however, this is not equivalent to an optimization over the submodular upper polyhedron. 
That is, we may not assume $c=0$ in Eqn.~\eqref{concaveextaff} for a submodular function.
Moreover, this expression requires solving a linear program over the submodular upper polyhedron, and it follows from Theorem~\ref{NPgenantisubpoly} that obtaining the concave extension is NP hard. We shall revisit this result in the next subsection while investigating the distributional characterization.

\subsection{Concave upper and lower bounds of the concave extension}
\label{sec:concave-upper-lower}

Interestingly, we can define a number of concave extensions based on
relaxations of the polyhedral representation. In particular, consider
the inner approximations of the generalized submodular upper
polyhedron $\mathcal P^f_{\mathcal G, \text{gen}}$, defined via a particular
set of supergradients 
$\mathcal G = \{ g_X \in \partial^f(X) | X \subseteq V \}$.
Instead of minimizing over all affine upper bounds,
we can minimize only over a particular class of modular upper
bounds. Then, we can define the following form of a concave extension:
\begin{align}\label{concaveextsupergrad}
\cex f_{\mathcal G}(w) 
\triangleq \min_{(y, c) \in \mathcal P^f_{\mathcal G, \text{gen}}} [\langle y, w \rangle + c] = \min_{Y \subseteq V} \,\ [\langle y, g_Y \rangle + f(Y) - g_Y(Y)], \,\,\ \forall w \in [0, 1]^{|V|}
\end{align}
In particular, the above turns the linear program into a discrete
optimization problem. Moreover, the concave extension $\cex f_{\mathcal G}$ is
guaranteed to be an upper bound of $\cex f$. We can define three
variants of these extensions using the 
polytopes $\growsymb{\mathcal G} = \{
\ggrow_X | X \subseteq V \}$, $\shrinksymb{\mathcal G} = \{ \gshrink_X
| X \subseteq V \}$, and $\barsymb{\mathcal G} = \{ \gbar_X | X
\subseteq V \}$
and which we call $\cex f_{\growsymb{\mathcal G}}$, $\cex f_{\shrinksymb{\mathcal G}}$, and
$\cex f_{\barsymb{\mathcal G}}$. These concave extensions
can, in fact, be obtained in polynomial time since it involves
submodular function minimization for each evaluation.

The class of concave extensions suggested by
Eqn.~\eqref{concaveextsupergrad} has some connections
to a form of concave extension proposed
in~\cite{vondrak2007submodularity} for monotone submodular
functions. In particular, where \cite{vondrak2007submodularity}
defined a concave function $\cex f_g$ that takes the following form:
\begin{align}
\cex f_{\mathcal G_\text{v}}(x) 
\triangleq \min\{[f(Y) + \sum_{j \in V} x(j) f(j | Y)] | Y \subseteq V\}
= \min\{[f(Y) + \sum_{j \notin Y} x(j) f(j | Y)] | Y \subseteq V\}
\end{align}
This extension can be seen as a special case of Eqn.~\eqref{concaveextsupergrad} with a particular set of supergradients 
$\mathcal G_\text{v} = \{ g_X \in \partial^f(X) | X \subseteq V \}$
defined as:
\begin{align}
g_X(j) = 
\begin{cases}
0 & \text{ if }  j \in X\\
f(j |X) & \text { if } j \notin X\\
\end{cases}
\end{align}
This supergradient is related to the supergradient $\gshrink_X$ 
in Eqn.~\eqref{eq:gshrink_def}
except that it replaces the
values $f(j | V \setminus j )$ 
for $j \in X$ with $0$. For a monotone
submodular function, this remains a
supergradient (but not for a non-monotone submodular function).
This form of concave extension is NP hard to evaluate~(see Section 3.7
in \cite{vondrak2007submodularity}) but is still useful in
obtaining approximate maximizers for certain special cases (see Section~\ref{concavemax}).

Using outer bounds of the generalized submodular upper polyhedron,
defined by considering only a subset of inequalities that define
$\mathcal P_{\text{gen}}^f$, it would be possible to define tractable
lower bounds on the concave extension. We leave this to future work
(see Section~\ref{sec:concl-open-probl}).

\subsection{Distributional characterization of the concave extension}
\label{sec:distr-char-conc}

As with the convex extension and as shown in
Section~\ref{sec:distr-char-conv}, an alternate and equivalent
characterization of the concave extension can be viewed through a
distributional lens.

\begin{lemma}\label{eqconcvext}
Recall from Eqn.~\eqref{eq:lambda_w}
the set $\Lambda_w$ defined here again for convenience:
\begin{align}
\Lambda_w \triangleq 
\Bigl\{
\{
\lambda_S, S \subseteq V \}: 
\sum_{S \subseteq V} \lambda_S 1_S = w, \sum_{S \subseteq V} \lambda_S = 1,
\text{ and } \forall S, \lambda_S \geq 0
\Bigr\}.
\tag{\ref{eq:lambda_w}}
\end{align}
The concave extension from Eqn.~\eqref{concaveextaff} then can also be represented as:
\begin{align}
\label{concaveextdist}
\cex f(w) = \max_{\lambda \in \Lambda_w} \sum_{S \subseteq V} \lambda_S f(S)
\end{align}
\end{lemma}
The proof of the above follows on similar lines as the convex extension, and is shown in~\cite{dughmi2009submodular}. Unfortunately, unlike the convex extension, this extension is NP hard to evaluate and optimize over. 
\begin{proposition}\label{NPccvext}
Given a submodular function $f$, it is NP hard to evaluate and optimize $\cex f$.
\end{proposition}
This result is shown in~\cite{vondrak2007submodularity}. 



Similar as for the polyhedral characterization, we can relax the
distributional characterization to consider specific simplified
distributions. In particular, we can obtain the multilinear extension,
through a particular distribution, namely $\{ \lambda_S = \prod_{i \in S} x_i
\prod_{i \notin S} (1 - x_i), S \subseteq V \} \in \Lambda_x$. 
Then the multilinear extension is
defined as follows:
\begin{equation}
\label{eq:multlinextdist}
\mex f(x) \triangleq 
\sum_{S \subseteq V} \lambda_S f(S) 
= \sum_{S \subseteq V} f(S) \prod_{i \in S} x_i \prod_{i \notin S} (1 - x_i)
\end{equation}
It is not hard to see that this forms a lower bound on the concave
extension $\cex f$. This extension is not concave, however, unlike the
extensions described in Section~\ref{sec:concave-upper-lower} that
are. Similar to the concave extension it is hard to evaluate this
extension, and typically requires
sampling~\cite{vondrak2007submodularity} in practice to get
an estimate, although special cases exist
where it can be analytically expressed and computed exactly~\cite{cvxframework}.

\subsection{Concave extensions of subclasses of submodular functions}
While the concave extension $\cex f(x)$ is NP hard to compute in general, it can be done efficiently for certain subclasses of submodular functions. These include, for example, sums of weighted matroid rank functions~\cite{vondrak2007submodularity}, and the class of $M^{\natural}$-concave functions~(c.f. Theorem 6.42 in \cite{murota2003discrete}).

\subsection{Concave extensions and submodular maximization}
\label{concavemax}
The concave extensions and the multilinear extension have interesting connections to submodular maximization.  The following lemma from~\cite{vondrak2007submodularity} connects many of these extensions:
\begin{lemma}\cite{vondrak2007submodularity}
For every monotone submodular function $f$, $\cex f(x) \geq \mex f(x) \geq (1 - \frac{1}{e})\cex f(y)$. 
\end{lemma}

It is also possible  to relate all of the three extensions of a submodular function, 
namely the convex extension, the concave extension, and the multilinear extension.
\begin{lemma}
Given a submodular function, it holds that
\begin{align}
\cex f(x) \geq \mex f(x) \geq \lex f(x)
\end{align}
\end{lemma}
\begin{proof}
  The proof of this result follows directly from the distributional
  characterization of the convex and concave extensions,
  Eqns.~\eqref{convexextdist}, \eqref{concaveextdist}, and
  \eqref{eq:multlinextdist}.  Note that the multilinear extension uses
  a particular distribution, the concave extension is a pointwise
  maximum over all such distributions, while the convex extension is a
  pointwise minimum over these distributions.
\end{proof}

The facts above were used in providing a relaxation based algorithm
for maximizing a subclass of submodular functions
efficiently~\cite{vondrak2007submodularity}. This relaxation based
algorithm maximizes the concave extension $\cex f(x)$ that, while
NP hard to optimize in general, can be maximized in certain special
cases. The particular special case which is considered in
\cite{vondrak2007submodularity} is the class of weighted matroid rank
functions for which the concave extension has a simple
form. Furthermore, a pipage rounding method ensures no
integrality gap with respect to the multilinear extension, thus
providing a $1 - \frac{1}{e}$ approximation algorithm for the problem
of maximizing a monotone submodular function subject to a matroid
constraint. Furthermore, later, a conditional gradient style
algorithm, also called the continuous greedy
algorithm~\cite{vondrak2008optimal}, directly optimizes the
multi-linear extension thereby providing a general $1 - 1/e$
approximation algorithm for monotone submodular maximization subject
to matroid constraints. This was later extended to the non-monotone
case by~\cite{chekuri2011submodular}.

\section{Optimality Conditions for submodular maximization}\label{submodmaxopt}
Just as the subdifferential of a submodular function provides optimality conditions for submodular minimization, the superdifferential provides the optimality conditions for submodular maximization. 

\subsection{Unconstrained submodular maximization}
In this section, we consider the general problem of unconstrained submodular maximization:
\begin{equation}
\max_{X \subseteq V} f(X)
\end{equation}
Given a submodular function, we can give KKT-like conditions for
submodular maximization, and this is done in the following theorem:
\begin{lemma}
For a submodular function $f$, a set $A$ is a maximizer of $f$, if $\mathbf{0} \in \partial^f(A)$.
\end{lemma}
However as expected, finding the set $A$, with the property above, or even verifying if for a given set $A$, $\mathbf 0 \in \partial^f(A)$ are both NP hard problems (from Lemma~\ref{NPsuperdiffmembership}). However thanks to submodularity, we show that the aforementioned outer bounds on the superdifferential provide approximate optimality conditions for submodular maximization. Moreover, unlike the superdifferential, these bounds are easy to obtain.

\begin{proposition}
\label{prop:local_max_poly}
For a submodular function $f$, if $\mathbf 0
\in \partial^f_{\symmdiff(1, 1)}(A)$ then $A$ is a local maxima of $f$
(that is, $\forall B \supseteq A$, $f(A) \geq f(B)$ and $\forall C
\subseteq A$, $f(A) \geq f(C)$). Furthermore, if we define $S =
\argmax_{X \in \{ A, V \setminus A \}} f(A)$, then $f(S) \geq
\frac{1}{3} OPT$ where $OPT$ is the optimal value.
\end{proposition}

The above result is interesting since a very simple outer bound on the
superdifferential leads us to an approximate optimality condition for
submodular maximization. The local optimality condition follows
directly from the definition of $\partial^f_{\symmdiff(1, 1)}(A)$ and
the approximation guarantee follows directly from Theorem 3.4
in~\cite{janvondrak}.

We can also provide a sufficient condition for the maximizers of a submodular function.
\begin{lemma}
\label{prop:global_max_inpoly}
  If for any set $A$, $\mathbf{0} \in \inpolygrowshrink(A)$,
  then $A$ is the global maxima of the submodular function. 
\end{lemma}
\begin{proof}
  This proof follows from the fact that $\inpolygrowshrink(A)
  \subseteq \partial^f(A) \subseteq \partial^f_{(1, 1)}(A)$.  Thus, if
  $\mathbf{0} \in \inpolygrowshrink(A)$, it must also belong to
  $\partial^f(A)$, which means $A$ is the global optimizer of $f$.
\end{proof}

Based on Lemmas~\ref{prop:local_max_poly}
and~\ref{prop:global_max_inpoly}, if a local maxima $A$ is found
(which is relatively easy because $\partial^f_{(1, 1)}(A)$ is easy to
characterize) and if it happens that $\mathbf{0} \in
\inpolygrowshrink(A)$ (which is easy to check since
$\inpolygrowshrink(A)$ is easy to characterize), then we have a
certificate of a global maxima of the submodular function.

\subsection{Constrained submodular maximization}
\label{sec:constr-subm-maxim}

We can also provide similar results for constrained submodular
maximization as we show in the present section.  We consider a
constrained submodular maximization problem with $\mathcal C \subseteq 2^V$
representing a set of sets, and consider the following problem:
\begin{equation}
\max_{X \in \mathcal C} f(X)
\end{equation}
For example $\mathcal C$ could represent a cardinality constraint $\{X
\subseteq V: |X| \leq m\}$, or a spanning tree, matching, s-t path
constraints, etc. Another common type of constraints are matroid
independence constraints (i.e., $\mathcal C$ consists of the set of
independent sets of some matroid).Denote $\mathcal I$ is the
independent set of a matroid $\mathcal M$. Then $\mathcal C = \{X: X
\in \mathcal I\}$ is a matroid constraint. Similarly $\mathcal C = \{X
\subseteq V: c(X) \leq B\}$ represents a knapsack constraint.  We
therefore refer to these as \emph{combinatorial constraints}.

We then define a constraint-cognizant modification $\partial^f_{\mathcal C}(A)$ of the
superdifferential (Eqn.~\eqref{supdiff-def}) as follows:
\begin{align}
\forall X \in \mathcal C \subseteq 2^V,\;\;
\partial^f_{\mathcal C}(X) \triangleq \{x \in \mathbb{R}^n: f(Y) - x(Y) \leq f(X) - x(X), \forall Y \in \mathcal C\}
\label{eq:constraint_cognizant_super}
\end{align}
In other words, we only consider the feasible sets associated with the constraints. Then we can trivially define a KKT like optimality condition for the optimization problem:
\begin{lemma}
For a submodular function $f$, a set $A$ is a maximizer of the problem $\max_{X \in \mathcal C} f(X)$, if $0 \in \partial^f_{\mathcal C}(A)$.
\end{lemma}
Clearly finding the set above is NP-hard. However, similar to the
unconstrained setting, we show that, in a number of cases,
approximating the superdifferential can lead to polynomial time
algorithms for constrained submodular maximization with worst case
approximation guarantees. This is done in several scenarios as we next
discuss.

\subsubsection{Constrained monotone submodular function maximization}
\label{sec:constr-monot-subm}

Consider here a case where $f$ is a monotone submodular function, and
$\mathcal C$ is the constraint that the set belongs to the
intersection of the independence sets of $k$ matroids. Let $\mathcal
M_1, \mathcal M_2, \cdots, \mathcal M_k$ represent the $k$ matroids,
with corresponding independence sets $\mathcal I_1, \mathcal I_2,
\cdots, \mathcal I_k$. Then $\mathcal C = \{X: X \in \cap_{i = 1}^k
\mathcal I_i\}$. Analogous to how we defined outer bounds
$\partial^f_{\symmdiff(k, l)}(X)$ on the superdifferential in
Section~\ref{sec:outer-bounds-superd} and
Eqn.~\eqref{eq:super_diff_outer_bounds}, we can also define the
outer-bounds $\partial^f_{\mathcal C, \symmdiff(k, l)}$ of
$\partial^f_{\mathcal C}$ that correspond to $\partial^f_{\symmdiff(k,
  l)}(X)$ but that also are restricted to $\mathcal C$ in the sense of
Eqn.~\eqref{eq:constraint_cognizant_super}. We then make the following
observation for the problem of monotone submodular maximization
subject to matroid constraints.
\begin{observation}
Given a monotone submodular function $f$ and a constraint set $\mathcal C = \cap_{i = 1}^k \mathcal I_i$, for any set $A \in \mathcal C$, the following holds:
\begin{enumerate}

\item If $\mathbf{0} \in \partial^f_{\mathcal C, \symmdiff (2,
    k+1)}(A)$, then $f(A)$ is guaranteed to be at least $\frac{1}{k+1}$
  times the optimal value. In particular, for the special case of
  monotone submodular maximization subject to a single matroid
  constraint, for any set $A \in \mathcal C$, if $\mathbf{0}
  \in \partial^f_{\mathcal C, (2, 2)}(A)$, then $f(A)$ is guaranteed
  to be at least $\frac{1}{2}$ times the optimal value;

\item If $k = 1$ and $\mathcal C$ is a cardinality (uniform matroid
  constraint $\{X: |X| \leq m\}$), for any $r > 0$, a set $A$
  satisfying $\mathbf{0} \in \partial^f_{\mathcal C, \symmdiff(r + 1, r+1)}$ is
  guaranteed to have an approximation guarantee no worse than
  $\frac{m}{2m - r}$.
\end{enumerate}
\end{observation}
The first part of the observation (i.e., point 1) follow directly from
Corollary 2.4 in~\cite{lee2009non}. In the case of $k = 1$, the same
result was shown in~\cite{fisher1978analysis}). Moreover, it was also
shown in~\cite{fisher1978analysis}, that for the problem of monotone
submodular maximization subject to $k > 1$ matroid constraints, the
approximate optimality conditions $\mathbf 0 \in \partial^f_{\mathcal C, \symmdiff(2,
  2)}(A)$, can be arbitrarily bad, thus requiring ``higher order''
optimality conditions. We also remark that when $r = 1$ (i.e.,
submodular maximization subject to a single matroid constraint), this
is the same approximation factor that can be obtained by the simple greedy
algorithm~\cite{nemhauser1978}. 

The second part (point 2) of the above observation (which is submodular
maximization subject to cardinality constraints) follows from Theorem~8 in \cite{filmus2013inequalities}. In the case when $r = 1$, the
condition $\mathbf{0} \in \partial^f_{\mathcal C, \symmdiff(2, 2)}$ provides a
guarantee of $m/(2m-1)$ which is a slight improvement of $1/2$ in the
special case of cardinality constraints. An interesting observation is
that with better forms of local optima (i.e., the condition $\mathbf{0}
\in \partial^f_{\mathcal C, \symmdiff(r + 1, r+1)}$, is a local optima up to
size $r$), imply better approximation guarantees to this problem.

The approximation factor for $k \geq 2$ matroids can actually be improved as shown in~\cite{matroidimproved}. 
\begin{observation}
  Given a maximization problem of a monotone submodular function $f$
  subject to $k > 1$ matroid constraints, for a set $A \in \mathcal
  C$, if $\mathbf{0} \in \partial^f_{\mathcal C, \symmdiff(p+1, kp+1)}(A)$,
  then $f(A)$ is guaranteed to be at least $\frac{1}{k+1/p}$ times the
  optimal value. In particular, for a special case of monotone
  submodular maximization subject to $2$ matroid constraints, for any
  set $A \in \mathcal C$, if $\mathbf{0} \in \partial^f_{\mathcal C,
    \symmdiff (p+1, 2p+1)}(A)$, then $f(A)$ is guaranteed to be at least
  $\frac{1}{2+1/p}$ times the optimal value.
\end{observation}
This is currently the best known result for $k > 1$ matroids, a result
that follows from Corollary 3.1 in \cite{matroidimproved}.


Overall, the main insight in these results is that the local optima which are
obtained through the local search algorithms can all be viewed as
(approximate) \emph{optimality conditions} obtained via outer bounds
of the superdifferential of a submodular function.


\subsubsection{Constrained non-monotone submodular function maximization}
\label{sec:constr-symm-subm}

Finally, we consider the case of non-monotone submodular maximization
subject to $k$ matroid constraints. We first consider the case of symmetric
submodular functions, i.e., $f(A)=f(V \setminus A)$ for all $A \subseteq V$.
\begin{observation}
Given a symmetric submodular function $f$,  
\begin{enumerate}

\item If the constraint set is as follows $\mathcal C = \cap_{i = 1}^k
  \mathcal I_i$, then any set $A \in \mathcal C$ satisfying $\mathbf{0}
  \in \partial^f_{\mathcal C, \symmdiff (2, k+1)}(A)$, $f(A)$ is guaranteed to be at least
  $\frac{1}{k+2}$ times the optimal value.


\item If $\mathcal C$ is the set of bases of a
  matroid, then any set $A$ satisfying $0 \in \partial^f_{\mathcal C, \symmdiff (2,
    2)}(A)$ is guaranteed to have a valuation at least $1/3$ of the
  optimal.

\end{enumerate}
\end{observation}
The results in this proposition follow directly from the definitions
above, and through the results in~\cite{lee2009non} (the first part
follows from Theorem 2.8, while the second part is implied by Theorem
5.1).

We lastly provide an approximation bound in terms of
superdifferentials for non-monotone submodular maximization.
\begin{observation}
Given a non-monotone submodular function $f$, and $\mathcal C$ is a cardinality (uniform matroid constraint $\{X: |X| \leq m\}$), for any $r > 0$, a set $A$ satisfying $\mathbf{0} \in \partial^f_{\mathcal C, (r + 1, r+1)}$ is guaranteed to have an approximation guarantee no worse than $\frac{r}{2m - r}$.
\end{observation}

This result follows directly from Theorem 8
in~\cite{filmus2013inequalities}. The best bounds for non-monotone
submodular maximization require running several iterations of local
search procedures. In particular, the procedure of \cite{lee2009non}
runs $k+1$ local search procedures to obtain a $1/(k + 2 + 1/k)$
approximation algorithm for non-monotone submodular maximization
subject to a single matroid constraint. When $k = 1$, running two
rounds of this local search procedure results in a $1/4$
approximation. The individual local search procedures here obtains a
set $A$ satisfying $\mathbf{0} \in \partial^f_{\mathcal C, (2,
  k)}(A)$.


\section{Concave Characterizations: Discrete Separation Theorem and
  Fenchel Duality Theorem}
\label{submodfdtdstmax}

In Section~\ref{submodfdtdstmin}, we investigated forms of the
Discrete separation theorem (DST), the Fenchel duality theorem, and
the Minkowski sum theorem for submodular functions and their
associated polyhedra when seen from the convex perspective. We here analyze
forms of the discrete separation theorem and Fenchel duality theorem
for submodular functions and their associated polyhedra from the
concave perspective.

\subsection{(Concave) Discrete Separation Theorem}
\label{sec:discr-seper-theor-concave}

We first show that a restricted version of a form of Discrete
separation theorem holds that in some sense is the opposite of the
Frank's DST shown in Section~\ref{sec:discr-seper-theor}
(Lemma~\ref{lemma:orig_dst}).  In particular, we see how the current
result is a form of ``concave-like'' variant of the Discrete separation theorem. This
shows that, under some very mild restrictions, given a submodular
function $f$ and a supermodular function $g$ with $f(X) \leq g(X),
\forall X$, there exists a modular function $h$ such that $f(X) \leq
h(X) \leq g(X)$. This therefore shows how submodularity can be seen as
analogous to concavity, in the same way that
Lemma~\ref{lemma:orig_dst} shows how submodularity can be seen as
analogous to convexity.  The lemma follows:

\begin{lemma}[Concave Discrete Separation Theorem (CDST)]
  Given a submodular function $f$ and a supermodular function $g$,
  such that $f(X) \leq g(X), \forall X \subseteq V$, such that
  either $f(\emptyset) = g(\emptyset)$ or $f(V) = g(V)$, there exists a
  modular function $h$ such that $f(X) \leq h(X) \leq g(X), \forall X
  \subseteq V$. Moreover, when $f$ and $g$ are integral (and satisfy
  the above conditions), there exists an integral $h$ satisfying the
  above.
\label{thm:cdst}
\end{lemma}
\begin{proof}
Assume first that $f(\emptyset) = g(\emptyset)$. Then Let $h(X) = f(\emptyset) + \sum_{j \in X} f(j | \emptyset)$. Then the following chain of inequalities hold:
\begin{align}
f(X) \leq h(X) = f(\emptyset) + \sum_{j \in X} f(j | \emptyset) \leq g(\emptyset) + \sum_{j \in X} g(j | \emptyset) \leq g(X),
\end{align}
which follows since $f(j | \emptyset) = f(j) - f(\emptyset) \leq g(j) - g(\emptyset) = g(j | \emptyset)$. The rest of the inequalities follow from submodularity (and supermodularity) of $f$ (and $g$). The result for when $f(V) = g(V)$ analogously follows by considering the functions $f(V \backslash X)$ and $g(V \backslash X)$ which are submodular and supermodular respectively.
\end{proof}
In particular, when $f$ and $g$ are normalized $f(\emptyset) =
g(\emptyset) = 0$, the above always
holds with a normalized (i.e., $h(\emptyset) = 0$) modular function. 
Lemma~\ref{thm:cdst}, however, can in fact be further
generalized. In particular, it is not hard to see that the result goes
through whenever $\argmin_X [g(X) - f(X)]$ is either $\emptyset$ or
$V$. To show this, we provide a generalized form
that depends on the outcome of $\argmin_X [g(X) - f(X)]$.

\begin{lemma}[Generalized Concave Discrete Separation Theorem]
  Given a submodular function $f$ and a supermodular function $g$,
  such that $f(X) \leq g(X), \forall X \subseteq V$, let $A \in
  \argmin_X [g(X) - f(X)]$. Then there exists a modular function $h$
  such that $f(X) \leq h(X) \leq g(X)$, $\forall X: X \subseteq A$ and $
  X \supseteq A$.
\label{thm:gcdst}
\end{lemma}
\begin{proof}
  The proof of this result is analogous to the earlier one. 
  First, define $\alpha = \min_X g(X) - f(X)$
  and  $A \in \argmin_X g(X) - f(X)$.
  Then $\alpha \geq 0$, and we
  can define function $f'$ with $f'(X) = f(X) + \alpha$ so that $f'(A)
  = g(A)$. Then given a modular separating function $h$ with $f'(X)
  \leq h(X) \leq g(X)$ for all $X$, we have that $f'(A) = h(X) = g(A)$, and that
  $f(X) \leq h(X) \leq
  g(X)$ for all $X$. Hence, given $A
  \in \argmin_X g(X) - f(X)$, we may and do assume without loss of generality that
  $f(A) = g(A)$. 

Next, define 
\begin{align}
h(X) = f(A) + \sum_{j \in X \backslash A}
  f(j | A) - \sum_{j \in A \setminus X} f(j | A \backslash j)
\end{align} 
By Lemma~\ref{thm:supergrad_subset_supersets},
we have that $\forall X \in [0, A] \cup [A,V]$, 
$f(X) \leq h(X) = f(A) + \tilde g_A(X) - \tilde g_A(A)$.
Moreover, since $f(A) = g(A)$, and applying
a supermodular variant of Lemma~\ref{thm:supergrad_subset_supersets},
we have that $\forall X \in [0,A] \cup [A,V]$:
\begin{align}
h(X) &= f(A) + \sum_{j \in X \backslash A} f(j | A) -
  \sum_{j \in A \setminus X} f(j | A \backslash j) \\
   & \leq g(A) + \sum_{j \in
    X \backslash A} g(j | A) - \sum_{j \in A \setminus X} g(j | A
  \backslash j) \leq g(X).
\end{align}
\end{proof}
Observe that Lemma~\ref{thm:cdst} is a special case of
Lemma~\ref{thm:gcdst} as when $f(\emptyset) = g(\emptyset)$ (resp.\ 
$f(V) = g(V)$) we have $\alpha = 0$ and we may take $A = \emptyset$ (resp.\ $A=V$).
The discrete separation theorems, however, might not hold under the
most general conditions on $f$ and $g$. However, they do hold for
certain important subclasses. For example, if $f$ is a $M^{\natural}$-concave
function (which is submodular when restricted to $2^V$), and $g$ is a
$M^{\natural}$-convex function (which is supermodular when restricted
to $2^V$), the discrete separation theorem always holds~(c.f Theorem
8.15 in \cite{murota2003discrete}).

\subsection{Superdifferential Fenchel Duality Theorem}
\label{sec:superd-fench-dual}

Finally, we show that a version of the Fenchel duality theorem also
holds in certain restricted cases. Given a submodular function $f$ (or
equivalently supermodular function $g$), define the concave Fenchel
dual functions $f_*$, and correspondingly $g_*$, as:
\begin{align}
f_*(y) = \min_{X \subseteq V} [ y(X) - f(X)], \qquad \text{ and } \qquad
 g_*(y) = \max_{X \subseteq V} [y(X) - g(X)].
\end{align}
The Fenchel duals $f_*$ and $g_*$ are concave and convex functions
respectively. Unlike the convex Fenchel duals, obtaining these
expressions exactly is NP hard, since they correspond to submodular
maximization. These can however be approximately obtained up to
constant factors. The following lemma gives a restricted version
of Fenchel duality theorem:
\begin{lemma}
Given a submodular function $f$ and a supermodular function $g$ such that the 
concave discrete separation theorem holds,
\begin{align}
\max_{X \subseteq V} f(X) - g(X) = \min_x g_*(x) - f_*(x)
\end{align}
Furthermore, if $f$ and $g$ are integral and satisfy the CDST, the
maximum on the right hand side is attained by an integral vector $x$.
\end{lemma}
\begin{proof}
The proof of this result follows directly from Theorem 4 of~\cite{fujishigenarayanan2005}. In particular, \cite{fujishigenarayanan2005} show that given set functions $f$ and $g$ such that the discrete separation theorem holds, the Fenchel duality theorem will also hold for this pair of functions $f$ and $g$.
\end{proof}
Unlike the Fenchel duality theorem from the convex perspective, the
result above does not hold in the fully general setting. Moreover, if
the functions $f$ and $g$ are $M^{\natural}$-concave and
$M^{\natural}$-convex respectively, the Fenchel duality theorem always
holds~(c.f Theorem 8.21 in \cite{murota2003discrete}).

\subsection{Superdifferential Minkowski sum theorem}
\label{sec:superd-mink-sum}

Analogous to the results above, we show a certain restricted form of
the Minkowski sum theorem.
\begin{lemma}
Given two submodular functions $f_1$ and $f_2$, it holds that that\footnote{Recall, the addition of the polyhedra corresponds to point-wise addition}:
\begin{align}
\mathcal P^{f_1 + f_2} = \mathcal P^{f_1} + \mathcal P^{f_2}
\end{align}
Similarly, $\partial^{f_1 + f_2}(\emptyset) = \partial^{f_1}(\emptyset) + \partial^{f_2}(\emptyset)$ and $\partial^{f_1 + f_2}(V) = \partial^{f_1}(V) + \partial^{f_2}(V)$.
\end{lemma}
\begin{proof}
  This result follows directly by considering definitions. In
  particular, one can see that the extreme points of these polyhedra
  can be explicitly characterized by a submodular function. For
  example, the polyhedron $\mathcal P^{f_1}$ has a single extreme
  point defined by the vector $f_1(j), j \in V$. Similarly, the
  extreme point of $\mathcal P^{f_2}$ is $f_2(j), j \in V$, and the
  extreme point of $\mathcal P^{f_1 + f_2}$ is $f_1(j) + f_2(j), j \in
  V$, and hence the Minkowski sum theorem holds.
\end{proof}
Unlike the Minkowski sum theorem on the subdifferential and submodular
polyhedron, the result above may not hold for the superdifferential $\partial^f(X)$ of
an arbitrary set $X\subseteq V$, nor must it hold for the generalized submodular upper
polyhedron $\mathcal P_{\text{gen}}^f$. They do hold, however, for
certain subclasses of submodular functions, such as (once again) the class of
$M^{\natural}$-concave functions. This fact follows from
Theorem 3 in \cite{fujishigenarayanan2005}, and the fact that
$M^{\natural}$-concave functions satisfy the Fenchel duality theorem.

\section{Conclusions and Future Work}
\label{sec:concl-open-probl}

In this manuscript, we investigated several connections between convex
and concave aspects of submodular functions. We provided
characterizations of the superdifferentials, concave extensions and
separation and duality theorems related to concave aspects of a
submodular function, and connected these new results to existing
results on the convex aspects of submodular functions. To our
knowledge, this is the first work in this direction. We also show how
for specific subclasses of submodular functions, such as the class of
$M^{\natural}$-concave set functions, this characterization is exact,
while for other submodular functions, this can be done approximately.

We lastly discuss a few problems that remain open and that could be
considered for future work.
\begin{itemize}

\item Are there are other subclasses of submodular functions (apart
  from the class of $M^{\natural}$-concave set functions) for which
  the concave aspects, like the superdifferentials, concave extensions
  and characterizations like the discrete separation theorem, Fenchel
  duality theorem, and so on, can be provided exactly. In particular, we saw
  that the $M^{\natural}$-concave set functions satisfy the property
  that $\partial^f(X) = \partial^f_{\Delta(2, 2)}(X)$. An interesting
  question is whether there are other interesting subclasses of
  submodular functions exist that satisfy similar conditions on their
  superdifferential (e.g., $\partial^f(X) = \partial^f_{\Delta(k,
    k)}(X)$ for constant $k$).  Characterizing such functions could
  lead to exact polytime in $|V|$ (although exponential in $k$)
  algorithms for maximizing such subclasses of submodular functions.

\item In section~\ref{submodmaxopt}, we investigated optimality
  conditions related to submodular maximization and its connection to
  the superdifferential. An interesting open problem is if this
  characterization could provide insight into algorithms for
  submodular maximization, and conditions when submodular maximization
  can be done exactly. Moreover, it also is interesting that
  approximating the superdifferential provides different approximation
  algorithms for submodular maximization. It will be interesting if
  there is a principled relationship between these two.

\item In Section~\ref{submodfdtdstmax}, we study the Fenchel duality
  theorem, the discrete separation theorem, and the Minkowski sum theorem. We
  show that these results hold under restricted settings. An open
  question is if Edmonds intersection theorem (cf. Section 4.1
  in~\cite{fujishige2005submodular}) also holds under certain
  restricted settings.

\item In Section~\ref{sec:inner-outer-bounds} we defined inner bounds
  on the generalized submodular upper polyhedron and mentioned how it
  would also be easy to define outer bounds on this polyhedron. Then
  in Section~\ref{sec:polyh-char-conc} we used these inner bounds to
  produce upper bounds on the concave extension of a submodular
  function. An open question is if the aforesaid outer bounds, which
  would produce corresponding tractable lower bounds on the concave
  extensions, would be useful for optimization or certain
  applications.

\item Figure~\ref{fig:2d_superdifferential_figure} geometrically
  suggests that a rotation of the axes could transform the set of
  superdifferentials into a set of subdifferentials. It would be
  interesting to see if in the general case, for arbitrary size $V$,
  if some multi-axis rotation might perform a similar rotation, and to
  understand the complexity of identifying this rotation (something
  that at least must be NP-hard do the hardness of submodular
  maximization).

\item It may be interesting to consider the Lasserre
  hierarchy\cite{rothvoss2013lasserre}, and how it might be related to
  the complexity of the inequalities needed to characterize the
  superdifferential and bound the complexity of computing submodular
  maximization for certain submodular functions.

\rishabh{Jeff}{
 things to look at:
  - Sherali, Adams 
  - Lasserre hierarchy
    \url{http://www.math.washington.edu/~rothvoss/lecturenotes/lasserresurvey.pdf}
  this seems to be related to tree-width, but is it somehow
  related to the complexity of the inequalities needed to characterize
  the superdifferential (and thus bound the complexity of computing
  submodular maximization exactly)??
}
\rishabh{Jeff}{
I added a bullet for this, see above}

\item Additional analytical expressions of set functions are
  $M^{\natural}$-concave when the family of sets is restricted to a
  Laminar family of subsets of $2^V$ \cite{murota2008submodular}. It
  may be elucidating to consider how the superdifferential, and in
  particular $\partial^f_{\symmdiff(2, 2)}(X)$, relates to Laminar
  families.



\item Finally, thanks to the Minkowski sum theorem, the \lovasz{}
  extension of a submodular function satisfies that $\lex{f}_{1+2}(x)
  = \lex{f_1}(x) + \lex{f_2}(x)$, where $f_{1+2}(x) = f_1(x) +
  f_2(x)$, i.e., the \lovasz{} extension of a sum of two submodular
  functions is equal to the sum of the individual \lovasz{}
  extensions. An open problem is whether this relation holds (under
  restricted settings possibly) for the various concave extensions
  defined in this manuscript.
\end{itemize}

\rishabh{Jeff}{add text here on the open problems that we talked about, many of which are mentioned in the comments above.}\jeff{Rishabh}{I added a paragraph above. Let me know how this sounds.}\rishabh{Jeff}{I've added four more paragraphs, and edited everything.}

\rishabh{Jeff}{AUGTODO: double check that all of the references are
  correct, and clean, and full, and that we've not misspelled anyone's
  name or anything like that. This is important!}

\bibliographystyle{abbrv}      
\bibliography{../Combined_Bib/submod}
\end{document}